%% file: main.tex
\documentclass[12pt]{article}

\usepackage{amsmath} 
\usepackage{algorithm}
\usepackage{algpseudocode}
\usepackage{amsfonts}   
\usepackage{amssymb}  
\usepackage{pifont}
\usepackage{amsthm}
\usepackage{bbm}
\usepackage{tikz}
\usepackage{siunitx}
\usepackage{multirow}
\usepackage{graphicx}   
\usepackage{hyperref}   
\usepackage{authblk}    
\usepackage[margin=0.75in]{geometry} 
\usepackage[backend=biber, style=numeric,maxbibnames=9, sorting=none]{biblatex}

\usepackage{booktabs}   
\newcommand{\indep}{\perp \!\!\! \perp}
\theoremstyle{definition}
\newtheorem{definition}{Definition}[section]

\newtheorem{proposition}{Proposition}[section]

\newcommand{\innerproduct}[2]{\left\langle #1, #2 \right\rangle}
\title{Quantization-based LHS for dependent inputs : application to sensitivity analysis of environmental models.}
\author[1,2]{Guerlain Lambert\thanks{Email: \texttt{guerlain.lambert@ec-lyon.fr}}}
\author[1]{Céline Helbert\thanks{Email: \texttt{celine.helbert@ec-lyon.fr}}}
\author[2]{Claire Lauvernet\thanks{Email: \texttt{claire.lauvernet@inrae.fr}}}
\affil[1]{Institut Camille Jordan, CNRS UMR 5208, École Centrale de Lyon, Écully, France}
\affil[2]{RIVERLY, INRAE, Villeurbanne, France}
\date{\today}

\begin{document}

\maketitle

\begin{abstract}
Numerical modeling is essential for comprehending intricate physical phenomena in different domains. To handle complexity, sensitivity analysis, particularly screening, is crucial for identifying influential input parameters. Kernel-based methods, such as the Hilbert Schmidt Independence Criterion (HSIC), are valuable for analyzing dependencies between inputs and outputs. Moreover, due to the computational expense of such models, metamodels (or surrogate models) are often unavoidable. Implementing metamodels and HSIC requires data from the original model, which leads to the need for space-filling designs. While existing methods like Latin Hypercube Sampling (LHS) are effective for independent variables, incorporating dependence is challenging. This paper introduces a novel LHS variant, Quantization-based LHS, which leverages Voronoi vector quantization to address correlated inputs. The method ensures comprehensive coverage of stratified variables, enhancing distribution across marginals. The paper outlines expectation estimators based on Quantization-based LHS in various dependency settings, demonstrating their unbiasedness. The method is applied on several models of growing complexities, first on simple examples to illustrate the theory, then on more complex environmental hydrological models, when the dependence is known or not, and with more and more interactive processes and factors. The last application is on the digital twin of a French vineyard catchment (Beaujolais region) to design a vegetative filter strip and reduce water, sediment and pesticide transfers from the fields to the river. Quantization-based LHS is used to compute HSIC measures and independence tests, demonstrating its usefulness, especially in the context of complex models.
\end{abstract}
\noindent
{\bf Keywords:}  design of experiments, vector quantization, Latin hypercube sampling (LHS), HSIC
\vfill

\newpage
\section{Introduction}
\input{introduction}

\section{Existing tools : LHS, LHSD and Voronoi Vector Quantization}\label{section:2}
\input{section1}

\section{Quantization-based LHS}\label{section:3}
\input{section2}

\section{Application to HSIC : kernel-based sensitivity analysis}\label{section:4}
\input{section3}

\section{Numerical experiments}\label{section:5}
\input{section4}

\section{Conclusion}
\input{conclusion}

\section*{Acknowledgments}
This research was carried out with the support of the \href{https://doi.org/10.5281/zenodo.6581216}{CIROQUO} Applied Mathematics consortium, which brings together partners from both industry and academia to develop advanced methods for computational experimentation. The research is part of Water4All's AQUIGROW project, which aims to enhance the resilience of groundwater services under increased drought risk. The project has received funding from the European Union's Horizon Europe Program under Grant Agreement 101060874. 

\printbibliography
\end{document}

%% file: introduction.tex

Numerical models are used to represent and understand complex physical phenomena in fields such as in biology, geophysics, or hydrology, where processes are highly interactive. Such models can be complicated (e.g., black-box models), expensive, and difficult to use when, for example, the goal is to derive a digital twin for a specific real-world context. Therefore, it may be beneficial to replace the model with a metamodel, or surrogate model, which can be defined as a statistical model of the process-based model, and that is less costly to run. To deal with  numerous input parameters, it is also helpful to perform a two-steps global sensitivity analysis  to understand which inputs have a significant impact on the outputs \cite{saltelli_global_2008}, starting by a screening that allows “eliminating” a large set of inputs that are poorly influential on the studied outputs. For this screening step, Kernel-based methods, such as Hilbert Schmidt Independence Criterion (HSIC) \cite{arthur2005} which measures the dependence between inputs and outputs are very useful and have been implemented for a wide range of applications, with scalar data, vector, functional \cite{redatest}, or even sets \cite{fellmann}. The implementation of a metamodel relies on observations of the original code, which can be costly to obtain. Likewise, implementation of HSIC requires a sample of evaluations of the heavy code. Therefore, it is essential to create  a design of experiments (DOE) that is both small and fills as much space as possible. This is commonly referred to as a space-filling design. Yet, the physical reality of models often imposes dependence between input variables. One example is hydrology, where soil moisture is governed by the Van Genuchten equations \cite{VG}. The parameters of this model depend on the soil type, creating a set of dependent variables. From a design space filling point of view, it's therefore essential to offer a solution that takes dependence into account.
In the case of independent variables several space-filling designs, such as Latin Hypercube Sampling (LHS) \cite{LHS},  and low-discrepancy sequences, in particular Sobol sequences \cite{sobol_distribution_1967}  are commonly used. LHS are often preferred  \cite{Rouzies_gmd_2021} since they ensure space-filling properties, allowing accurate estimation of metamodels and good marginal covering stable after dimension reduction. Besides, they provide good properties for estimation of expectation, as required for HSIC measure computation.  Extensions of LHS to account for dependency have been made by \cite{iman_small_1980,stein_large_1987,Mondal_2020}, using methods based on ranks and copulas, the latter requiring knowledge of copulas and quantile functions not always available in practice. Alternatively, kernel-based methods such as kernel herding \cite{chen2012supersamples} consist in minimizing a squared Maximum Mean Discrepancy (MMD) between an iteratively-built sequence of points and the target correlated joint distribution. These deterministic approaches strongly depend on the choice of the kernels and introduce a bias in estimation of expectations, such as HSIC. There is therefore an interest in providing a ready-to-use method that requires a minimum of assumptions.

In this paper, we introduce a new LHS method based on Voronoi vector quantization (VQ) to take into account correlated inputs. \cite{saka_latinized_2007} proposes a design of experiments based on Voronoi VQ in the manner of LHS, with a Latinization procedure involving ranking. Later, \cite{CorlayPagès} is the VQ to stratification by randomly drawing $M$ points per Voronoi strata to apply it to functional data. The main difficulties with these methods are that they do not take dependency into account. Our new DOE, called quantization-based LHS, is a direct extension of Latin Hypercube sampling based on Voronoi quantization to take into account dependence within a group of input variables. It has good properties because it swaps the bias resulting from the use of Voronoi centroids for variance by randomly drawing a point in the Voronoi cells. This ensures complete coverage of the group of dependent variables being stratified. Combined with random permutations, we get a well-distributed design across all the marginals: the independent part and the dependent part. All that is required is how to simulate its distribution (and thus conditionally on the Voronoi cells).

To this end, in Section \ref{section:2}, existing LHS technics and Voronoi vector quantization are briefly recalled, with special reference to optimal quantization. Then, in Section \ref{section:3}, the contribution of this paper is presented. That is, expectation estimation using LHS in the context of dependent random variables based on vector quantization. Different estimators are proposed to cover three different settings:  a unique group of dependent inputs, a joint distribution between a group of dependent inputs independent to another input, a joint distribution between two independent,  groups of dependent variables. In particular, it is shown that the proposed estimators are unbiased. Finally, quantization-based LHS is applied to the computation of HSIC measures and independence tests in Section \ref{section:4}.
To illustrate the relevance of these developments, they are tested and compared to other methods (Monte Carlo and LHSD) on two operational environmental models in  Section \ref{section:5}: (i) a flood risk model where the dependency structure  between inputs and the marginal laws are perfectly known and (ii) a chain of models that simulates water, sediment and pesticide transfers, where dependency is unknown. This last application is implemented on the digital twin of a real vineyard catchment in France (Morcille experimental site), where  pollution of agricultural origin by pesticides is a recognized public health problem \cite{Morcille_datapaper}.

%% file: section1.tex
\subsection{Latin Hypercube Sampling and Latin Hypercube Sampling with dependence}
\subsubsection*{Latin Hypercube Sampling (LHS)}
The objective of an LHS of size $N$, as introduced by \cite{LHS}, is to sample $N$ points uniformly in $[0,1]^d$ such that marginally the projected points are well spread on $[0,1]$. The support of each coordinate, i.e. $[0,1]$, is partitioned into $N$ sub-intervals of equal size $\frac{1}{N}$. The points are drawn in $[0,1]^d$ by associating  each coordinate to one sub-interval and by uniformly sampling in this sub-interval. The LHS procedure is as follows :

\begin{algorithm}
\caption{LHS}
\label{alg:LHS}
\begin{algorithmic}[1]
\State Generate $N$ independent samples $(U_{i1}, \dots, U_{id})_{i = 1, \dots, N}$, where $U_{ij}$ is i.i.d $\mathcal{U}([0,1])$.
\State Generate $d$ independent equiprobable permutations $\pi_1, \dots, \pi_d$ of $\{1, \dots, N\}$. $\pi_{j}(i)$ is the value to which $i$ is mapped by the $j^{\text{th}}$ permutation.
\State An LHS is given by:
\begin{align*}
\displaystyle \begin{cases}
V_{ij} \ =\ \frac{\pi _{j}(i) -1}{N} \ +\ \frac{U_{ij}}{N}\\
j\ \ =\ 1,\dotsc ,d\ \ ,\ i\ =\ 1,\dotsc ,\ N\ 
\end{cases}\end{align*}
\end{algorithmic}
\end{algorithm}
We aim to estimate $\mathbb{E}[f(X)]$ where $X \sim \mathcal{U}([0,1]^d)$, $d\in\mathbb{N}^*$. Given an LHS $(V_{i1},\dots, V_{id})_{1\leq i\leq N}$ of $\mathcal{U}([0,1]^d)$, the LHS estimator of the expected value is :
\begin{align*}
    \mu_{LHS} = \frac{1}{N}\sum_{i=1}^N f(V_{i1},\dots, V_{id})
\end{align*}
It is unbiased, and $Var(\mu_{LHS}) \leq \frac{N\times Var(\mu_{MC})}{N-1}$, this means that using an LHS of size $N$ will not result in any more variance than using a Monte Carlo sample of size $N$, see \cite{owen1997monte}. We also have access to a Central Limit theorem, see \cite{owen1992central}. It stratifies the marginal distribution to maximize coverage of the range of each variable. However, the use of LHS imposes the independence of random variables, which raises an issue if the model inputs are correlated. Therefore, it is crucial to use an appropriate experimental design. Over the years, several modifications to LHS have been proposed to incorporate dependence. For instance, \cite{iman_small_1980} introduced a rank-based approach, further improved by \cite{stein_large_1987}, but this results in a biased estimator. More recently, \cite{Mondal_2020} introduced a copula-based LHS method, LHSD, which enables the incorporation of dependence into the experimental design while retaining the properties of the LHS.
\subsubsection*{Latin Hypercube Sampling with Dependence (LHSD)}
Adding knowledge of the dependency structure through copulas, \cite{Mondal_2020} proposes an extension of LHS and \cite{stein_large_1987}, the LHSD, by constructing the joint distribution from known marginal distributions.  The LHSD procedure proposed by \cite{Mondal_2020} is based on the copula (and conditional copulas) construction summarized in Algorithm \ref{alg:LHSD}. Consider a random vector $X = (X_1, ..., X_d)$ of joint c.d.f $F$ where the marginal c.d.f are denoted $F_j$ for all $j=1, \dots, d$ and $C$ a copula. In particular, $C$ is a distribution with uniform marginals. Sklar's theorem allows a copula to be associated with any multidimensional distribution. The copula makes it possible to model the dependence between the marginals, if $F$ is continuous and let $F_j(X_j) =: U_j$, then :
\begin{align*}
    C(U_1, \dots, U_d) = F\left(F_1^{-1}(U_1),\dots,F_d^{-1}(U_d)\right)
\end{align*}
The conditional copula is defined by : 
\begin{align*}
    C_j(U_j\mid U_1,\dots,U_{j-1}) = \frac{\partial^{j-1}C(u_1,\dots, u_j,1, \dots,1)}{\partial u_1 \dots\partial u_{j-1}}
\end{align*}

\begin{algorithm}
\caption{LHSD from \cite{Mondal_2020}}
\label{alg:LHSD}
\begin{algorithmic}[1]
\State Generate an LHS sample $(Z_{i1}, \dots, Z_{id})_{i = 1, \dots, N}$ of $[0,1]^d$ using Algorithm \ref{alg:LHS}.
\State Construct sequentially $(U_{i1}, \dots, U_{id})_{i = 1, \dots, N}$ from the joint copula using inverse conditional copula functions:
     \begin{align*}
         U_{ij} = C_j^{-1}(Z_{ij} \vert U_{i1}, \ldots, U_{ij-1})
     \end{align*}
     for $j=1\dots d$.
\State Construct the final sample : $(X_{i1}, \dots, X_{id})_{i = 1, \dots, N}$ using the inverse distribution function $X_{ij} = F_j^{-1}(U_{ij})$ for $j = 1, \dots, d$.
\end{algorithmic}
\end{algorithm}
The properties of the LHSD estimator are similar to those of the LHS as justified in \cite{Mondal_2020}. Specifically, when estimating an expected value, the variance of LHSD is lower than the variance from a simple Monte Carlo sample.\\

The aim is to reconstruct the dependence of the marginals sequentially by constructing a sample on $[0,1]$ from the inverse of the conditional copula, starting from an LHS sample. To obtain a LHSD, apply the quantile transformation to this sample.\\

The proposed LHSD method requires knowledge of a copula and the inverse of conditional copulas. Although the copula can be estimated from known families (Gaussian, Clayton, ...), the practical implementation can be quite difficult depending on the complexity of the model and the dependency structure. To avoid these limitations, we propose a space-filling method based on vector quantization to preserve the dependency with few requirements.

\subsection{Background on vector quantization}
Vector quantization was first introduced in signal processing during the 1950s as a method of discretizing continuous signals. It is now widely used in various applications, including speech recognition \cite{rabiner1993fundamentals}, image compression \cite{3776}, and numerical probability \cite{pagès2018numerical}. The latter is of particular interest to us. Consider a probability space $(\Omega, \mathcal{A},\mathbb{P})$.\\

Let $X$ be a random vector with values in $\mathbb{R}^d$ with $\mathbb{P}_X$ its distribution and $N \in \mathbb{N}^*$. Let $\Gamma = \{x_1, \dots, x_N\}\subset \mathbb{R}^d$, the Voronoi partition associated to $\Gamma$ is defined as $(C_i)_{i=1\dots N}$ such that
\begin{align*}
    C_i(\Gamma) \subset \left\{y\in\mathbb{R}^d \hspace{1pt}:\hspace{1pt} \vert y-x_i\vert \leq \min_{1\leq j\leq N} \vert y - x_j\vert\right\}
\end{align*}
The Voronoi quantizer can be defined as follows: 
\begin{align*}
    q_{vor}(X) = \sum_{i=1}^N x_i \mathbbm{1}_{C_i(\Gamma)}(X)
\end{align*}
The quantized variable $q_{vor}(X)$ is often denoted as $\widehat{X}$. It is the discrete version of $X$ with the $N$ support points $x_1, \dots, x_N$. Similarly, the quantized probability distribution $\mathbb{P}_{\widehat{X}}$ of $\mathbb{P}_X$ induced by $\Gamma$ is given by : 
\begin{align*}
    \mathbb{P}_{\widehat{X}} = \sum_{i=1}^N \mathbb{P}_X\left(C_i(\Gamma)\right)\delta_{x_i}
\end{align*}
where $\delta_{x_i}$ is the Dirac mass centered on $x_i$.\\

To assess the performance of the $\widehat{X}$ quantizer, we introduce the quadratic distortion function associated to $\Gamma = \{x_1, \dots, x_N\}$ : 
\begin{align*}
    \mathcal{D}_N^X(\Gamma) = \lVert d(X, \Gamma) \rVert_{L^2(\mathbb{P})}^2 =  \mathbb{E}\left[\min_{1 \leq i \leq N} \vert X - x_i\vert^2\right] = \int_{\mathbb{R}^d} \min_{1\leq i \leq N} \vert x_i - y\vert^2 \mathbb{P}_X(\mathrm{d}y)
\end{align*}
Any quantizer that minimizes distortion is called an optimal quantizer. If $X\in L^2(\mathbb{P})$, then the existence of such quantifiers is assured. However, uniqueness is not systematically obtained (1D case of unimodal distributions) \cite{GrafVQ, pagès2018numerical}. Furthermore, any $N$-optimal quantizer $\widehat{X}$ is a stationary quantizer i.e :
\begin{align*}
    \widehat{X} = \mathbb{E}\left[X\hspace{1pt}\vert\hspace{1pt}\widehat{X}\right]
\end{align*}
In practice, this property enables the construction of fixed-point algorithms to obtain optimal (or at least suboptimal) quantization. The most well-known of these algorithms is Lloyd's algorithm, also known as k-means. It is very popular and easy to implement, especially if the distribution of $X$ is known from a large sample of simulated points. An example of optimal quantization is given in figure \ref{fig:fig1}, based on the centered bivariate normal distribution $X = (X_1, X_2)$ with $cov(X_1, X_2) = 0.8$. 
\begin{figure}[htbp]
    \centering
    \includegraphics[width=0.45\textwidth]{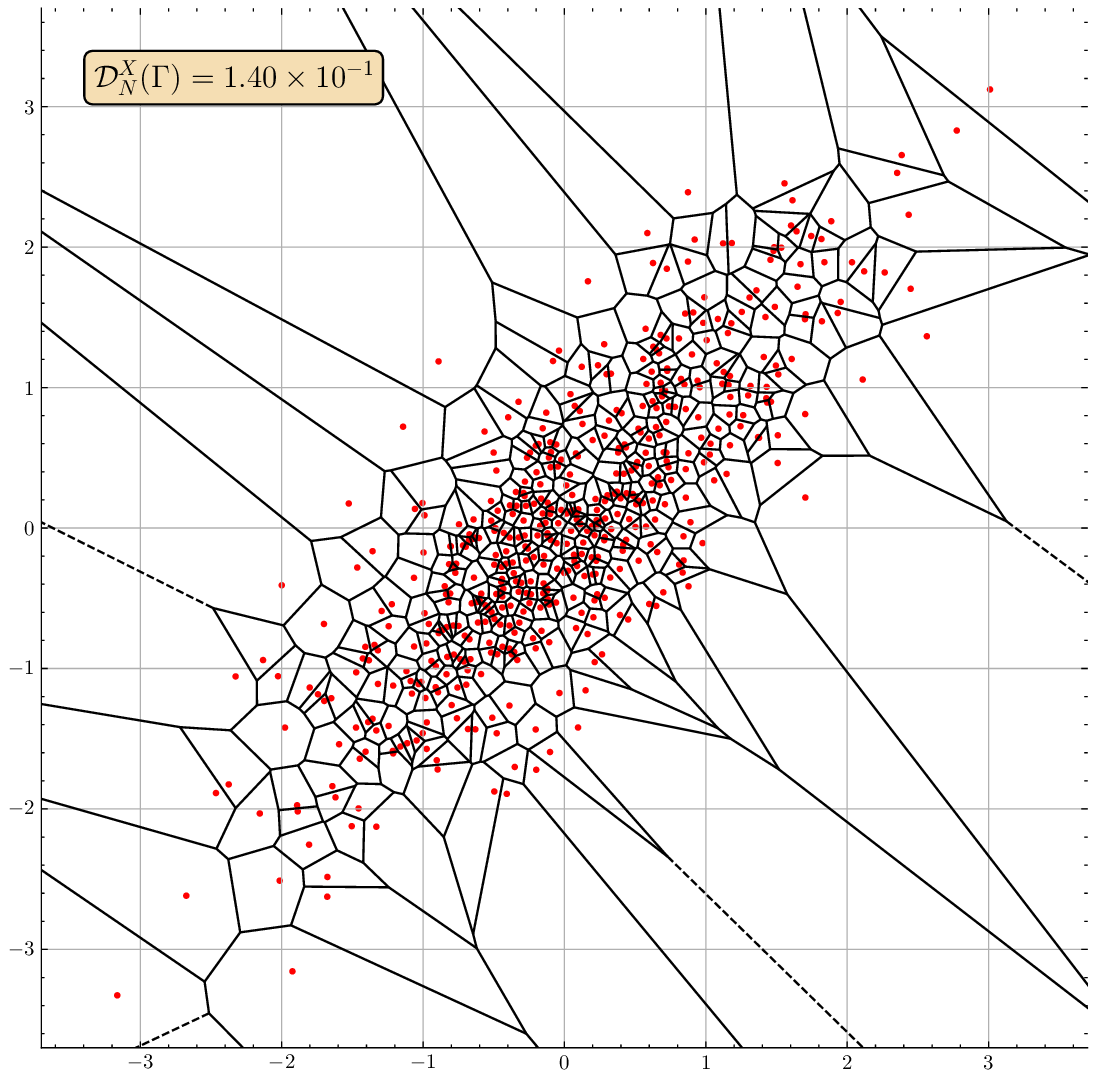}
    \hfill
    \includegraphics[width=0.45\textwidth]{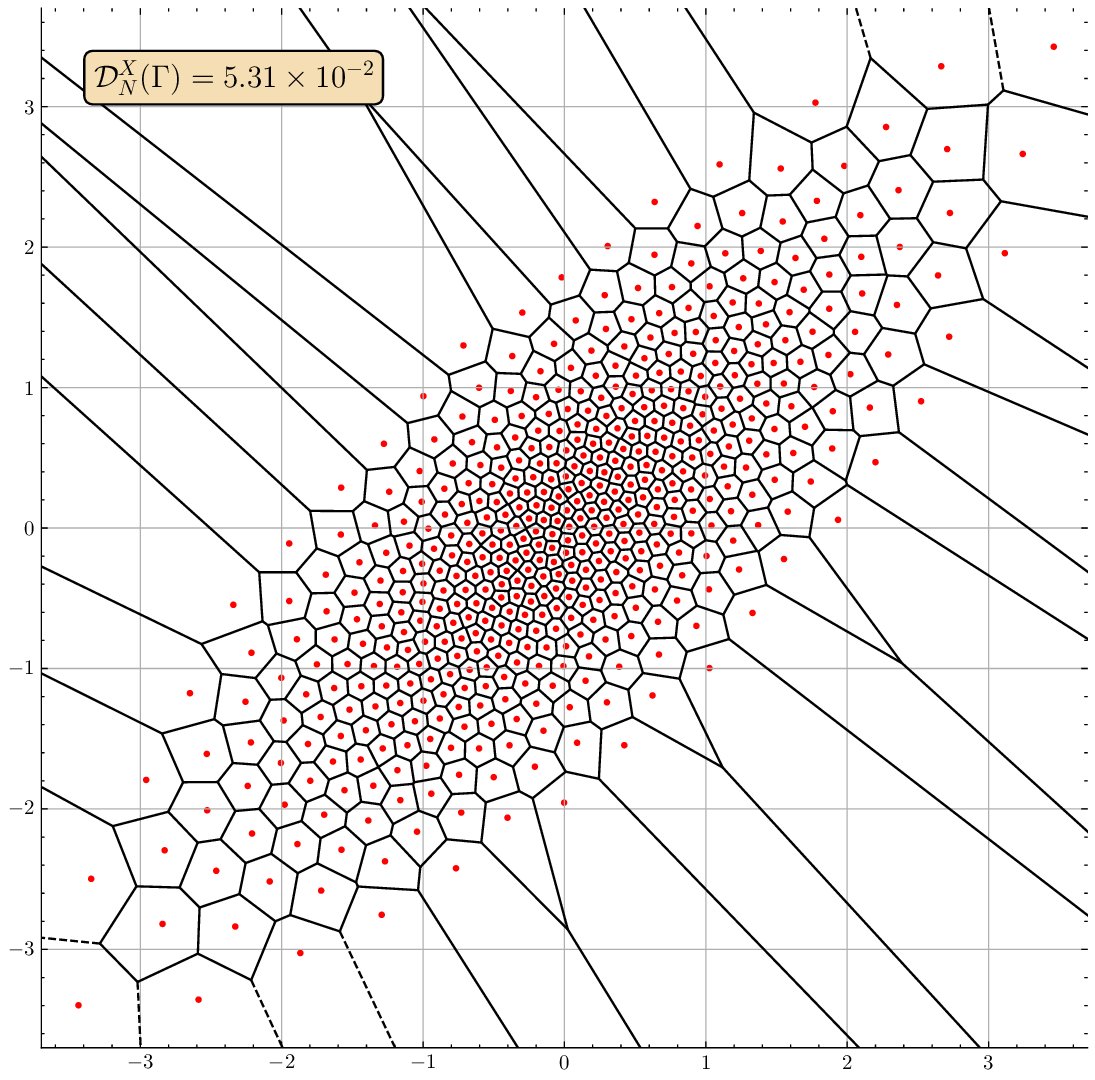}
    \caption{Voronoi quantization of the centered bivariate, with covariance of 0.8 between marginals encoded into $N = 500$ centroids. On the left : Voronoi tessellation of 500 random points from the joint distribution. On the right, Voronoi tessellation of 500 centroids obtained via optimal quantization (kmeans). }
    \label{fig:fig1}
\end{figure} 

In the context of vector quantization, cubature formulas are available to compute $\mathbb{E}[f(X)]$ where $X$ is a random vector on $\mathbb{R}^d$ and $f : \mathbb{R}^d \to \mathbb{R}$ a function \cite{pagès2018numerical}: 
\begin{align*}
    \mathbb{E}[f(X)] &\approx \mathbb{E}[f(\widehat{X})]
     = \sum_{i=1}^{N}f(x_{i})\mathbb{P}\left[\widehat{X} = x_{i}\right]
     = \sum_{i=1}^{N}f(x_{i})\mathbb{P}\left[X \in C_{i}\right]
\end{align*}
The accuracy of this estimate depends on the regularity of $f$. If $f$ is $L$-Lipschitz, then 
\begin{align*}
    \bigg\vert \mathbb{E}\left[f(X)\right] - \mathbb{E}\left[f\left(\widehat{X}\right)\right]\bigg\vert\leq L\sqrt{\mathcal{D}_N^X(\Gamma)}
\end{align*}
If $f \in \mathcal{C}^1(\mathbb{R}^d,\mathbb{R})$, $\nabla f$ is $L$-Lipschitz, and $\widehat{X}$ is stationary, then :
\begin{align*}
    \bigg\vert \mathbb{E}\left[f(X)\right] - \mathbb{E}\left[f\left(\widehat{X}\right)\right]\bigg\vert\leq \frac{L}{2}\mathcal{D}_N^X(\Gamma)
\end{align*}
Unlike Monte Carlo estimation, vector quantization estimation is deterministic. The estimator is therefore variance-free, but biased, whereas Monte Carlo estimation is unbiased, but with variance. In design of experiment (DOE), variance is preferred over bias. This leads to modification of estimation through vector quantization with the addition of controlled randomness, such as in stratification. Thus, one contribution of the paper is to introduce a stochastic version of vector quantization. This procedure will be used for the group of dependent inputs and associated to independent inputs in an LHS-style. This new methodology called quantization-based LHS is discussed in the following section.

%% file: section2.tex
In this section, we introduce  different sampling strategies which account for dependency and such that the space-filling property is maintained after dimension reduction, so they can be used for screening purposes. For each sampling strategy, we study the estimation of $m = \mathbb{E}[f(X)]$, for $f : \mathbb{R}^d \to \mathbb{R}$.  
Three cases are studied:
\begin{itemize}
\item case 1 : $m = \mathbb{E}[f(X)]$ where $X = (X_{1}, \dots, X_d)$ with $d$ dependent components 
\item case 2 : $m = \mathbb{E}[f(X,Y)]$ where $X = (X_1, \dots, X_s)$ is composed of $s$ dependent components, $Y = (Y_{1}, \dots, Y_{d-s})$ is composed of $d-s$ independent components and $X$ and $Y$ are independent.
\item case 3 : $m = \mathbb{E}[f(X,Y)]$ where $X = (X_1, \dots, X_s)$ is composed of $s$ dependent components, $Y = (Y_{1}, \dots, Y_{d-s})$ is composed of $d-s$ dependent components and $X$ and $Y$ are independent.
\end{itemize}

\subsection{Random quantization (case 1)}\label{sec:RQ}
In this section $X$ is composed of a unique group of $d$ dependent components $X_{1}, \dots, X_d$. We introduce a stochastic version of vector quantization, called Random Quantization (RQ) to account for dependency. RQ is a stratification technique. The strata are the Voronoi cells obtained after quantization. A point is randomly drawn in each cell according to the probability distribution of X conditional on the cell. The Random Quantization sampling method is summarized in Algorithm \ref{alg:RQ}.

\begin{algorithm}
\caption{Random Quantization sampling}
\label{alg:RQ}
\begin{algorithmic}
\State Let $X \in \mathbb{R}^d$ be a random vector composed of $d$ dependent components.
\State Let $(x_{i1}, \dots, x_{id})_{i = 1, \dots, N}$ be a $N$-optimal quantizer of $X$. Let $(C_i)_{i = 1, \dots, N}$ be the associated Voronoi partition of  $\mathbb{R}^d$.
\For{$i = 1$ to $N$}
    \State Generate one random point $U_i =$ $(U_{i1}, \dots, U_{id})$ in the cell $C_i$ according to the probability distribution of $X$ conditioned on $C_i$, i.e. $U_i \sim \mathcal{L}(X \hspace{2pt}\vert \hspace{2pt} X\in C_i)$.
\EndFor
\State \textbf{return} $U = (U_1, \dots, U_N)$
\end{algorithmic}
\end{algorithm}

The use of randomized vector quantization for DOE has the advantage that it only requires knowledge of how to simulate in Voronoi cells. No knowledge of copulas or quantile functions is necessary. This allows for the entire distribution of $X$ to be explained using a finite number of support points, while perfectly accounting for input dependence.\\
\begin{definition}
    Let $(U_i)_{i=1 \dots N}$ a sample provided by Algorithm \ref{alg:RQ}. We define the following RQ estimator : 
    \begin{align}
    \mu_{RQ} := \widehat{\mathbb{E}[f(X)]}_{RQ} = \sum_{i=1}^{N}f(U_i)\mathbb{P}[X\in C_{i}].
\end{align}

\end{definition}

\begin{proposition}
    The estimator $\mu_{RQ}$ is unbiased, and its variance is given by : 
    \begin{align*}
        \mathrm{Var}\left(\mu_{RQ}\right) = \sum_{i=1}^{N}\mathbb{P}[X\in C_i]^2\mathrm{Var}\left(f(U_i)\right)
    \end{align*}
\end{proposition}
\begin{proof}
    For the bias :
    \begin{align*}
        \mathbb{E}[f(X)] & = \sum_{i=1}^N\mathbb{E}\left[\mathbbm{1}_{X\in C_i}f(X)\right]\\
        & = \sum_{i=1}^{N}\mathbb{P}[X \in C_i]\mathbb{E}\left[f(X) \hspace{2pt}\vert \hspace{2pt} {X \in C_i}\right]\\
        & =  \sum_{i=1}^{N}\mathbb{P}[X \in C_i]\mathbb{E}\left[f(U_i)\right]\\
        & = \mathbb{E}\left[\mu_{RQ}\right]
    \end{align*}
    For the variance :
    \begin{align*}
         \mathrm{Var}\left(\mu_{RQ}\right) &= \sum_{i=1}^{N}\mathbb{P}[X\in C_i]^2 \mathrm{Var}\left(f(U_i)\right) + 2\sum_{1\leq i < j \leq N}\mathbb{P}[X\in C_i]\mathbb{P}[X\in C_j]\mathrm{Cov}\left(f(U_i), f(U_j)\right)\\
         & = \sum_{i=1}^{N}\mathbb{P}[X\in C_i]^2\mathrm{Var}\left(f(U_i)\right)
    \end{align*}
\end{proof}

We illustrate the behavior of $\mu_{RQ}$ on two toy examples. We first  consider the function $f : \mathbb{R} \to \mathbb{R}$ defined for all $x$ $\in \mathbb{R}$ by $f(x) = x^2$ and look for $\mathbb{E}[f(X)]$ with $X \sim \mathcal{N}(0,1)$. In the context of using a costly numerical model, the number of model evaluations is limited, therefore we choose low values of $N$ (10, 20, 50, and 100) and the behavior of $\mu_{RQ}$ is studied through 1000 repetitions. The results, summarized in Figure \ref{fig:fig2}, show that the estimator is unbiased and that its variance decreases as $N$ increases. $\mu_{RQ}$ is compared to Monte Carlo and LHS estimation. We recall that to obtain an LHS sample we first apply Algorithm \ref{alg:LHS} to produce a sample in $[0,1]$ to which the inverse of standard Gaussian c.d.f is applied. $\mu_{RQ}$ exhibits lower variance at low $N$, making it the most efficient method in this case.
\begin{figure}[htbp]
    \centering
    \includegraphics[width=\textwidth]{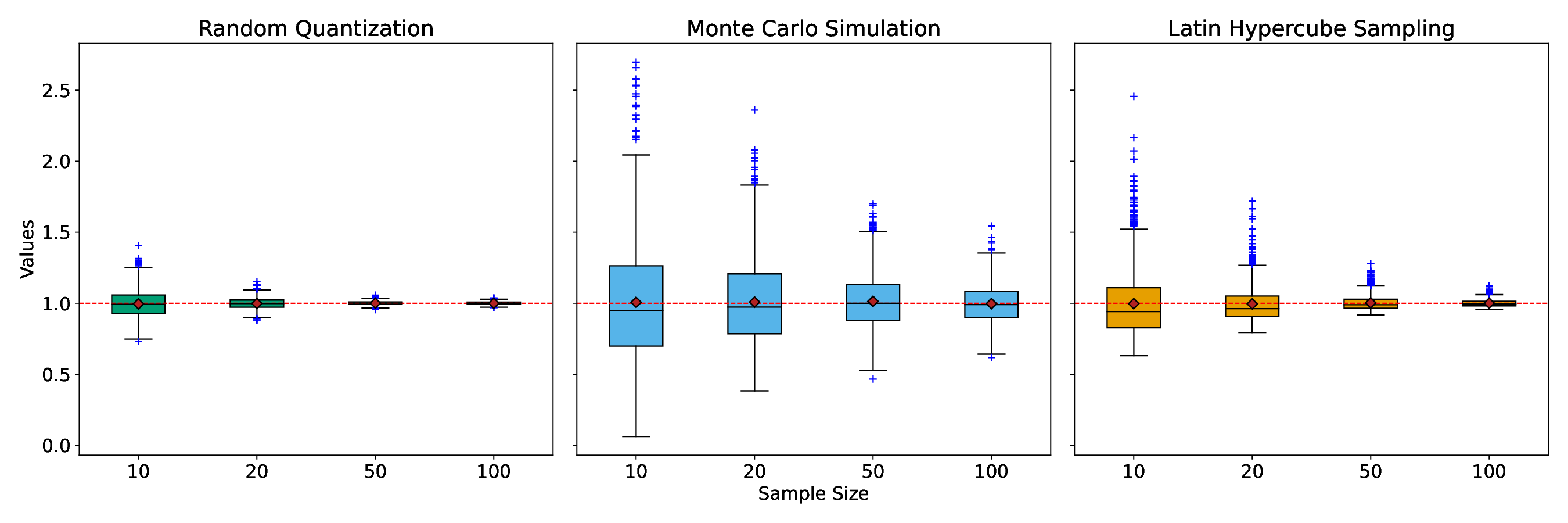}
    \caption{Estimation of $\mathbb{E}(X^2)$ where $X\sim\mathcal{N}(0,1)$ with sample size $N =10, 20, 50, 100$ and 1000 repetitions per $N$. The red dashed line is the theoretical value.}
    \label{fig:fig2}
\end{figure}
{We consider a second 2D example with correlation between marginals. Let $X = (X_1, X_2)$ be a centered Gaussian vector such that $cov(X_1, X_2) = 0.8$. To estimate $\mathbb{E}(X_1X_2)$, Monte Carlo, LHSD and Random Quantization were used. The results are shown in Figure \ref{fig:fig3}. A Gaussian copula with a correlation of 0.8 was used, and an analytical expression for the inverse of the conditional copula is known. All three estimators are unbiased, and the proposed estimator, $\mu_{RQ}$, has minimal variance.
\begin{figure}[htbp]
    \centering
    \includegraphics[width=\textwidth]{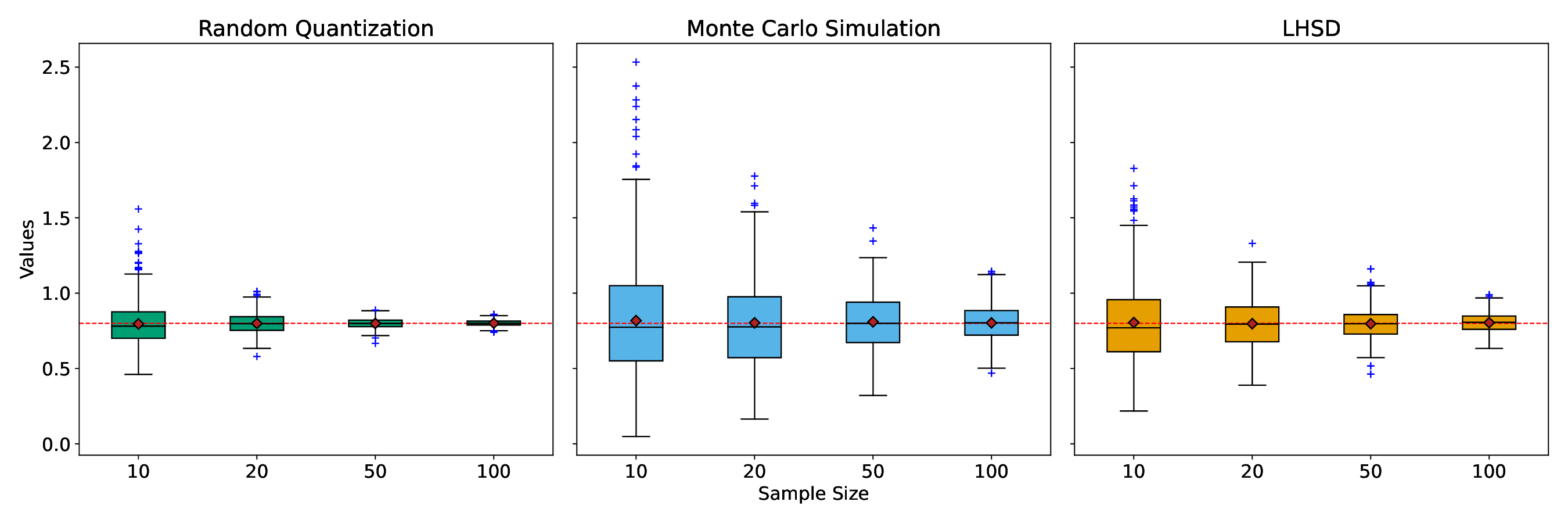}
    \caption{Estimation of $\mathbb{E}(X_1X_2)$ where $X=(X_1, X_2)$ is a centered Gaussian vector with covariance 0.8 with sample size $N \in \{10, 20, 50, 100\}$ and 500 repetitions per $N$. The red dashed line is the theoretical value of 0.8.}
    \label{fig:fig3}
\end{figure}
}

\subsection{Quantization-based LHS (case 2)}
Consider $X = (X_1, \dots, X_s)$, $s\in\mathbb{N}^*$, a random vector with dependent components and $Y=(Y_1, \dots, Y_{d-s})$, a group of independent random variables. $X$ and $Y$ are independent. We want to summarize the theoretical law of $(X,Y)$ though an empirical sample of $N$ points in $\mathbb{R}^d$ which preserves space-filling properties after dimension reduction. To do so, we stratify $X$ using the Random Quantification procedure introduced in section \ref{sec:RQ}. In the same way, we stratify $Y$ using an LHS sample. The idea is then to associate to each stratum of $X$ a stratum of $Y$ using a random permutation $\pi$. The sampling scheme is described in Algorithm \ref{alg:RQLHS}.

\begin{algorithm}
\caption{Quantization-based LHS}
\label{alg:RQLHS}
\begin{algorithmic}
\State Let $X \in \mathbb{R}^s$ be a random vector composed of $s$ dependent components.
\State Apply Algorithm \ref{alg:RQ} to provide a RQ sample $U = (U_1, \dots, U_N)$ of $X$.
\State Let $Y \in \mathbb{R}^{d-s}$ be a random vector composed of $d-s$ independent components.
\State Apply Algorithm \ref{alg:LHS} to provide an LHS sample $V = (V_1, \dots, V_N)$ of $Y$.
\State Let $\pi$ be a random permutation of $\{1,\dots, N\}$ in $\{1,\dots,N\}$.
\State \textbf{return} $((U_1,V_{\pi(1)}) , \dots, (U_N,V_{\pi(N)}))$
\end{algorithmic}
\end{algorithm}

\begin{definition}
Let $((U_i,V_{\pi(i)}))_{i=1 \dots N}$ a sample provided by Algorithm \ref{alg:RQLHS} where $U_i \sim \mathcal{L}(X | X\in C_i)$ and  $(C_i)_{i=1,\dots, N}$ is the Voronoi tessellation associated to the quantization of $X$. We define the following Quantization-based LHS estimator : 
   \begin{align}
    \mu_{QLHS} := \widehat{\mathbb{E}[f(X, Y)]}_{QLHS} = \sum_{i=1}^{N}\mathbb{P}[X\in C_{i}]f(U_{i}, V_{\pi(i)})
\end{align}
\end{definition}

\begin{proposition}
    The estimator $\mu_{QLHS}$ is unbiased.
\end{proposition}
In the following, for $i=1, \dots,N$, we define $p_i := \mathbb{P}[X\in C_i]$.
\begin{proof}
    For $N\in\mathbb{N}^*$, we denote by $S_N$ the permutation group of order $N!$. 
    \begin{align*}
        \mathbb{E}[\mu_{QLHS}] & = \sum_{i=1}^N p_i\mathbb{E}\left[f(U_i, V_{\pi(i)})\right]\\
        & = \sum_{i=1}^N p_i \mathbb{E}\left[\mathbb{E}\left[f(U_i, V_{\pi(i)}) \right]\hspace{2pt}\vert\hspace{2pt} \pi\right]\\
        &= \sum_{i=1}^N p_i \sum_{a\in S_N}\mathbb{P}[\pi = a]\mathbb{E}\left[f(U_i,V_{\pi(i)})\mid \pi = a\right]\\
        & = \sum_{i=1}^N p_i \sum_{a\in S_N}\frac{1}{N !}\mathbb{E}\left[f(U_i,V_{a(i)})\right]\\
        & = \sum_{i=1}^N \sum_{j=1}^N \sum_{\substack{a\in S_N\\ a(i)=j}} \frac{1}{N!}p_i \mathbb{E}\left[f(U_i, V_j)\right]\\
        & = \sum_{i=1}^N \sum_{j=1}^N \frac{(N-1)!}{N!}p_i \mathbb{E}\left[f(U_i, V_j)\right]\\
        & = \frac{N(N-1)!}{N!}\sum_{i=1}^N \sum_{j=1}^N \frac{1}{N}p_i \mathbb{E}[f(U_i, V_j)]\\
        & = \mathbb{E}\left[f(X, Y)\right]
    \end{align*}
    Because $\left\{a\in S_N \hspace{2pt}\vert\hspace{2pt} a(i)=j\right\} \cong S_{N-1}$. Hence, $\mathrm{Card}\left(\left\{a\in S_N \hspace{2pt}\vert\hspace{2pt} a(i)=j\right\}\right) = (N-1)!$.
\end{proof}
We study the behavior of $\mu_{QLHS}$ through the function $f : \mathbb{R}^2 \to \mathbb{R}$ defined for all $(x,y) \in \mathbb{R}^2$ by $f(x,y) = x^2y$. The problem is to estimate $\mathbb{E}[f(X,Y)]$ where $X \sim \mathcal{N}(0,1)$ and $Y \sim \mathcal{U}\left([0,1]\right)$. The results are summarized in Figure \ref{fig:fig4}, leading to the same conclusion as for $\mu_{RQ}$. Specifically, the estimator is unbiased and has decreasing variance, resulting in better performance than the typical Monte Carlo and LHS techniques. {By estimating $\mathbb{E}((X_1+X_2)^2Y)$ through the addition of dependence in $X = (X_1, X_2)$ from a centered bivariate Gaussian vector, we have obtained results shown in Figure \ref{fig:fig5}. The LHSD method is parameterized based on the Gaussian copula with parameter $\rho = 0.8$. It is observed that $\mu_{QLHS}$ is unbiased with lower variance than the MC and LHSD methods for all sample sizes, and therefore offers the best performance.}
\begin{figure}[htbp]
    \centering
    \includegraphics[width=\textwidth]{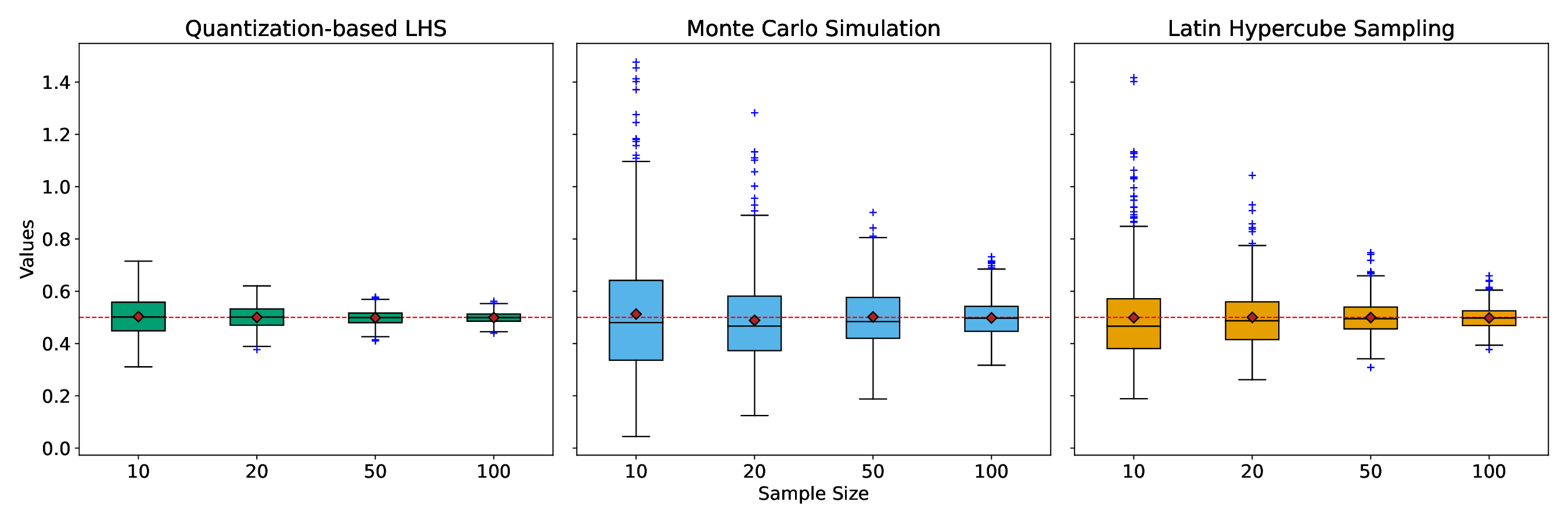}
    \caption{Estimation of $\mathbb{E}(X^2Y)$ where $X\sim\mathcal{N}(0,1)$ and $Y\sim \mathcal{U}([0,1])$ with sample size $N \in \{10, 20, 50, 100\}$ and 1000 repetitions per $N$. The red dashed line is the theoretical value.}
    \label{fig:fig4}
\end{figure}
\begin{figure}[htbp]
    \centering
    \includegraphics[width=\textwidth]{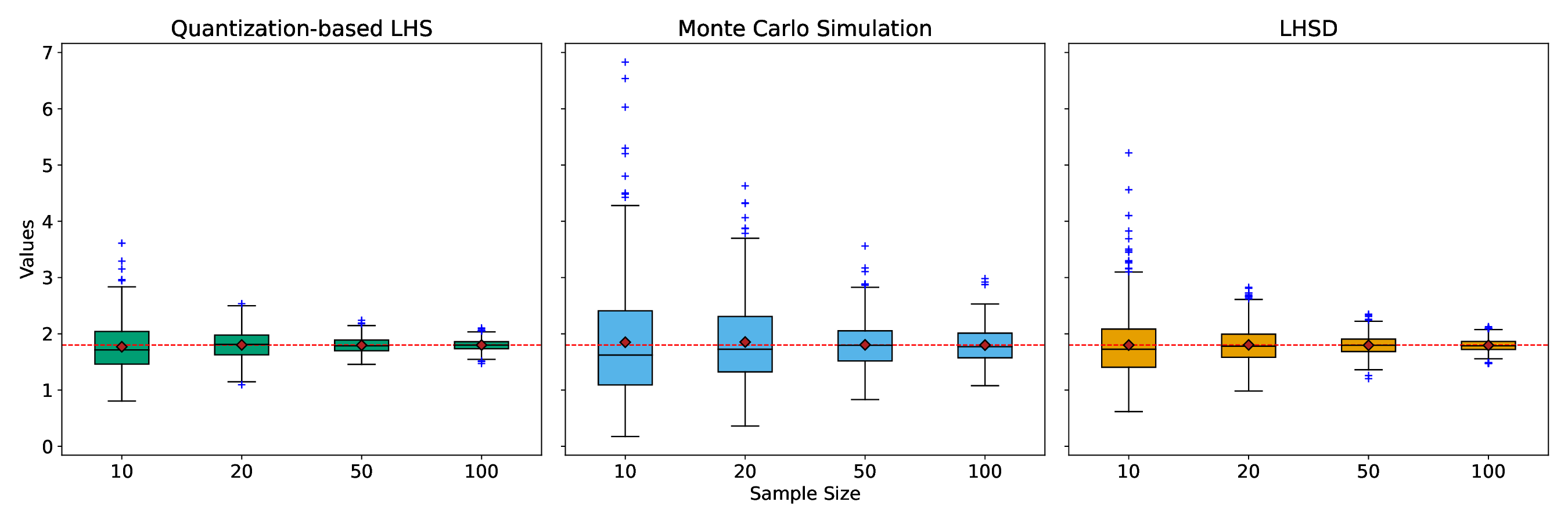}
    \caption{Estimation of $\mathbb{E}((X_1 + X_2)^2Y)$ where $X = (X_1, X_2)$ is a centered bivariate Gaussian vector with covariance $cov(X_1, X_2) = 0.8$ and $Y\sim \mathcal{U}([0,1])$ with sample size $N \in \{10, 20, 50, 100\}$ and 500 repetitions per $N$. The red dashed line is the theoretical value.}
    \label{fig:fig5}
\end{figure}
\subsection{Double Quantization-based LHS}
Considering two independent random vectors $X \in \mathbb{R}^s$ and $Y \in \mathbb{R}^{d-s}$, each composed of several dependent variables.  We want to compute $\mathbb{E}[f(X, Y)]$ where $f : \mathbb{R}^2 \to \mathbb{R}$ is a continuous function. Let consider two  samples $U = (U_1, \dots, U_N)$ and $V = (V_1, \dots, V_N)$ of $X$ and $Y$ obtained via Random Quantization from Algorithm \ref{alg:RQ}. These two samples preserve the correlation inside each group of dependent variables. In the same manner as in the previous section, the idea is to associate to each stratum of $X$ a stratum of $Y$ using a random permutation $\pi$. The sampling scheme is described in Algorithm \ref{alg:DRQLHS}.\\
\begin{algorithm}
\caption{Double Quantization-based LHS}
\label{alg:DRQLHS}
\begin{algorithmic}
\State Let $X \in \mathbb{R}^s$ be a random vector composed of $s$ dependent components.
\State Apply Algorithm \ref{alg:RQ} to provide a RQ sample $U = (U_1, \dots, U_N)$ of $X$.
\State Let $Y \in \mathbb{R}^{d-s}$ be a random vector composed of $d-s$ dependent components.
\State Apply Algorithm \ref{alg:RQ} to provide a RQ sample $V = (V_1, \dots, V_N)$ of $Y$.
\State Let $\pi$ be a random permutation of $\{1,\dots, N\}$ in $\{1,\dots,N\}$.
\State \textbf{return} $((U_1,V_{\pi(1)}) , \dots, (U_N,V_{\pi(N)}))$
\end{algorithmic}
\end{algorithm}
\begin{definition}
Let $((U_i,V_{\pi(i)}))_{i=1 \dots N}$ a sample provided by Algorithm \ref{alg:DRQLHS} where 
\begin{itemize}
    \item $U_i \sim \mathcal{L}(X | X\in C^X_i)$ and  $(C^X_i)_{i=1,\dots, N}$ is the Voronoi tessellation associated to the quantization of $X$
    \item $V_j \sim \mathcal{L}(Y | Y\in C^Y_j)$ and  $(C^Y_j)_{j=1,\dots, N}$ is the Voronoi tessellation associated to the quantization of $Y$
\end{itemize}
We define the following Q2LHS estimator : 
\begin{align}
\mu_{Q2LHS} =:\frac{1}{\sum_{i=1}^N p_iq_{\pi(i)}}\sum_{i=1}^N p_iq_{\pi(i)} f\left(U_i,V_{\pi(i)}\right)
\end{align}
where $\forall 1 \leq i \leq N$, $p_i =\mathbb{P}(X\in C^X_i) $ and $\forall 1 \leq j \leq N$, $q_j =\mathbb{P}(Y\in C^Y_j)$.
\end{definition}

Figure \ref{fig:fig6} shows 1000 iterations of the estimate of $\mathbb{E}(XY^2+Y^2)$ where $X\sim\mathcal{LN}(0,1)$ and $Y\sim\mathcal{N}(0,1)$ with its confidence interval. It is observed that the Q2LHS estimator is asymptotically unbiased.
Figure \ref{fig:fig7} compares the Q2LHS estimator with Monte Carlo and LHSD with $N \in \{10,20,50, 100\}$. The estimator is asymptotically unbiased and performs better than conventional Monte Carlo and LHS estimators, with smaller variance and fewer outliers.
\begin{figure}[htbp]
    \centering
    \includegraphics[width=0.6\textwidth]{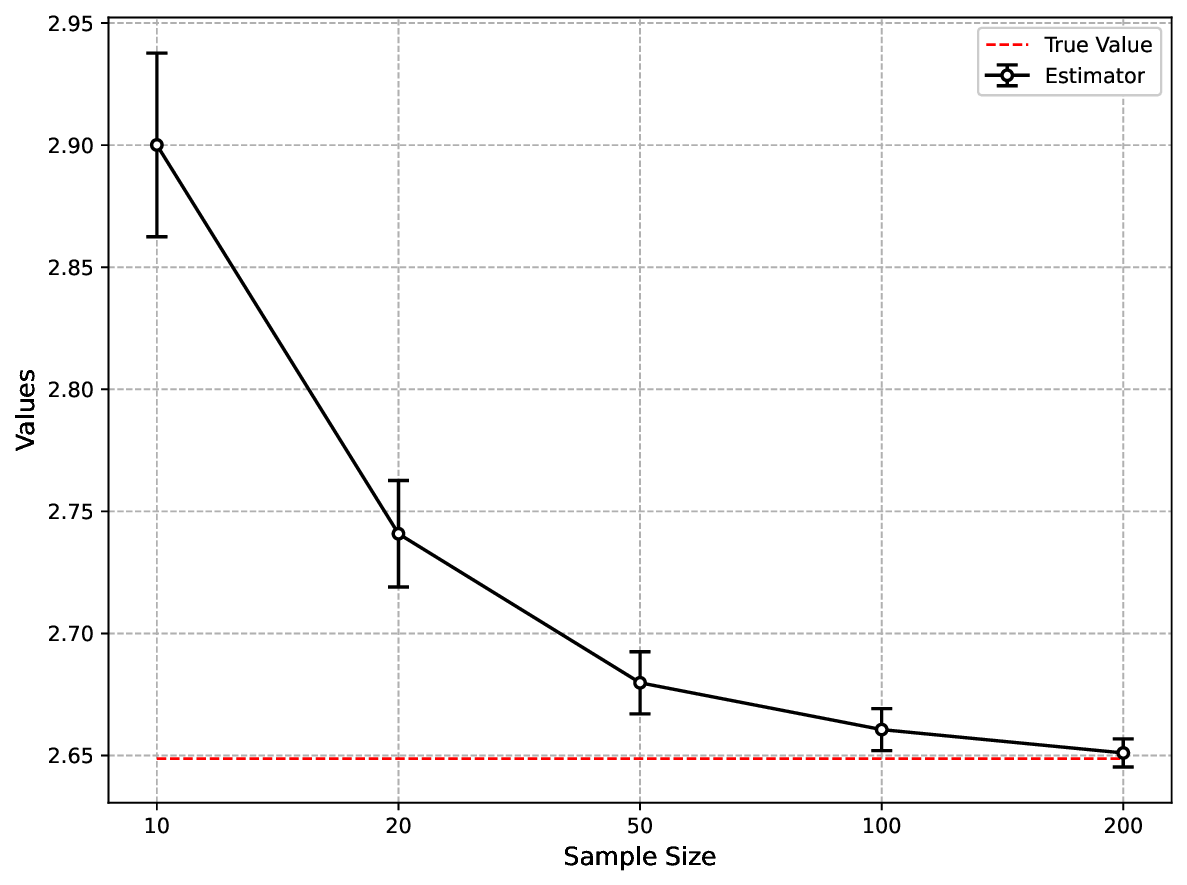}
    \caption{Estimation of $\mathbb{E}(\mu_{Q2LHS})$ with the sample size when estimating  $\mathbb{E}(XY^2+Y^2)$ where $X\sim\mathcal{LN}(0,1)$ and $Y\sim\mathcal{N}(0,1)$. 1000 repetitions of the estimation are performed with $N \in \{10,20,50, 100\}$. The 95\% confidence intervals are shown.}
    \label{fig:fig6}
\end{figure}

\begin{figure}[htbp]
    \centering
    \includegraphics[width=\textwidth]{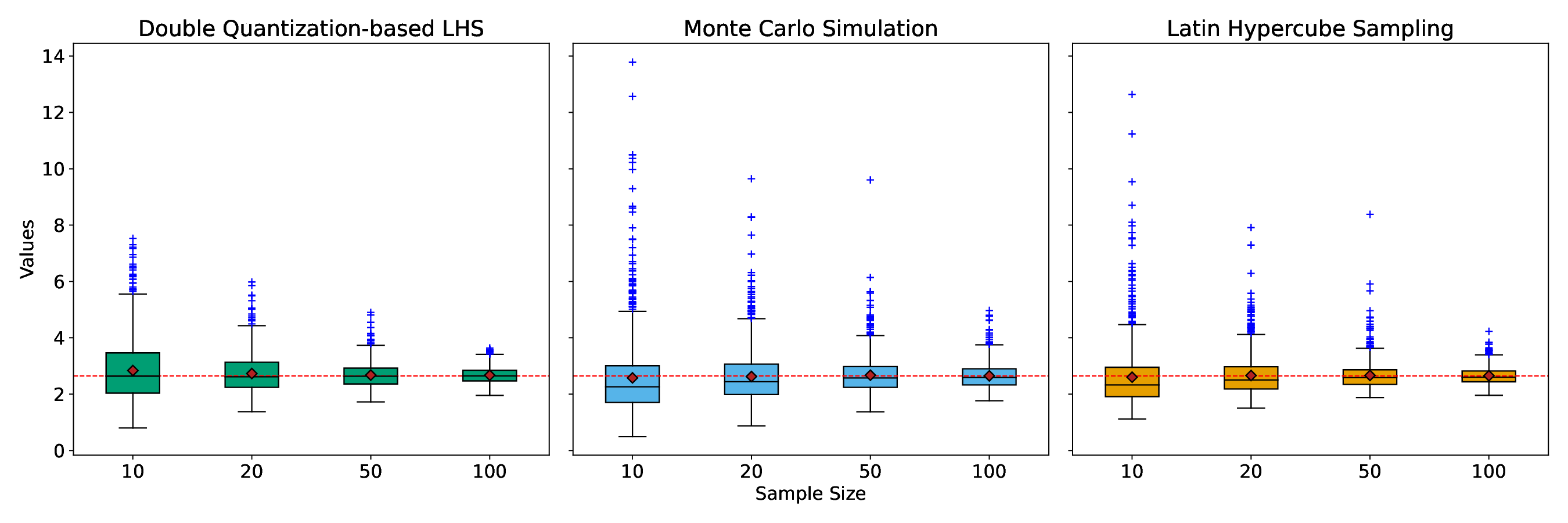}
    \caption{Estimation of $\mathbb{E}(XY^2 + Y^2)$ where $X\sim\mathcal{LN}(0,1)$ and $Y\sim\mathcal{N}(0,1)$ with sample size $N =10, 20, 50, 100$ and 1000 repetitions per $N$. The red dashed line is the theoretical value i.e $e^{0.5} + 1\approx 2.64$.}
    \label{fig:fig7}
\end{figure}

%% file: section3.tex
When building a surrogate model (or metamodel) for a numerical computation code, it is useful to carry out a preliminary selection of the most influential input variables.
This step is called screening and is crucial when the number of input parameters is large. It allows tuning the metamodel in a limited dimension space and thus requires a reduced number of costly evaluations of computer experiments. The goal of this section is to show how Quantization-based LHS allows estimating Sensitivity Analysis HSIC measures taking dependence into account.

\subsection{RKHS and HSIC measures}
Consider a numerical model $\mathcal{M}$ such that for $X\in\mathbb{R}^d$, $\mathcal{M}(X)\in\mathbb{R}$. Let $K \neq \emptyset$ and $\mathcal{H}$ be a Hilbert space of real functions in $K$.
\begin{definition}[Reproducing kernel]
    A kernel $k : K\times K \to \mathbb{R}$ of $\mathcal{H}$ is reproducing if we have : $\forall x\in K,\quad k(\cdot, x)\in\mathcal{H}$ and if it verifies the reproducing property :
    \begin{align*}
        \forall f\in\mathcal{H},\quad x\in K, \quad f(x) = \innerproduct{f}{k(\cdot,x)}
    \end{align*}
    The space $\mathcal{H}$ is said to be a Reproducing Kernel Hilbert Space (RKHS) if, for all $x\in K$, the Dirac function $\delta_x : \mathcal{H}\to K$ defined as 
    \begin{align*}
        \forall f\in\mathcal{H},\quad \delta_x(f) = f(x)
    \end{align*}
    is continuous.
\end{definition}
For more details on RKHS, the reader is referred to \cite{SVM}.
\begin{definition}[Kernel embedding]
    Let $\mathcal{M}^+_1$ the space of probability measures on $K$. Consider $\mathcal{H}$ the RKHS induced by a kernel $k : K \times K \to \mathbb{R}$. We define the kernel mean embedding as 
    \begin{align*}
        \mu : \begin{cases}
\mathcal{M}^+_1 \to \mathcal{H}\\
\mathbb{P} \mapsto \int k(\cdot, x)\mathrm{d}\mathbb{P}(x)
\end{cases}
    \end{align*}
\end{definition}
Consider $\mathbf{X} =(X_1, \dots ,X_d)$ a random vector defined on $\mathcal{X} = \mathcal{X}_1\times\dots\times\mathcal{X}_d$ and $Y$ a scalar output (which can be extended to vector or functional outputs) where $Y := \mathcal{M}(X)$ and $M : \mathcal{X} \to \mathbb{R}$ is a black-box numerical model. For a given set of indices $A\subset\{1,\dots, d\}$, we define the random vector $X_A$ as $(X_i)_{i\in A}$ on a probability space $(\Omega, \mathcal{A},\mathbb{P})$ with distribution $\mathbb{P}_{X_A}$.
\begin{definition}[HSIC measure]
Let $A \subset \{1, \dots, d\}$. Let $\mathcal{H}$ be the RKHS of functions of $\mathcal{X}_A$ in $\mathbb{R}$ with kernel $k:=\bigotimes_{i\in A} k_i$, and let $\mathcal{F}$ be the RKHS of functions of $\mathcal{Y}$ in $\mathbb{R}$ with kernel $k_Y$. The Hilbert-Schmidt independence criterion (HSIC) measures the distance between the embeddings of two distributions : the joint probability distribution $\mathbb{P}_{(X_A,Y)}$ of $(X_A, Y)$ and the product of the marginal probability distributions $\mathbb{P}_{X_A}$ and $\mathbb{P}_Y$ and is given by:
\begin{align*}
    \mathrm{HSIC}(X_A, Y) = \mathrm{MMD}\left(\mathbb{P}_{(X_A,Y)}, \mathbb{P}_{X_A}\otimes \mathbb{P}_Y\right)^2 = \left\Vert \mu\left(\mathbb{P}_{(X_A,Y)}\right) - \mu\left(\mathbb{P}_{X_A}\right)\otimes\mu\left(\mathbb{P}_Y\right)\right\Vert^2_{\mathcal{H}\times\mathcal{F}}
\end{align*}
where $\mu\left(\mathbb{P}_{(X_A,Y)}\right) = \mathbb{E}\left[k(X_A, \cdot)k_Y(Y, \cdot)\right]$ is the kernel mean embedding of the joint distribution, and $\mu\left(\mathbb{P}_{X_A}\right)\otimes\mu\left(\mathbb{P}_Y\right) = \mathbb{E}\left[k(X_A,\cdot)\right]\mathbb{E}\left[k_Y(Y,\cdot)\right]$ is the kernel mean embedding of the product of marginal distributions.
\end{definition}
We can note that HSIC is null if $X_A$ and $Y$ are independent and positive otherwise. It then allows selecting which input is influencing on the output. Besides, HSIC can be calculated for groups of random variables indexed by $A$. This is useful for groups of dependent variables. An illustrated example of the HSIC measure with RKHS is given in Figure \ref{fig:fig8}. The reproducibility property of RKHS allows us to derive an expression based exclusively on the expectations of the kernels, as summarized in the following proposition:
\begin{figure}[htbp]
    \centering
    \includegraphics[width=0.5\textwidth]{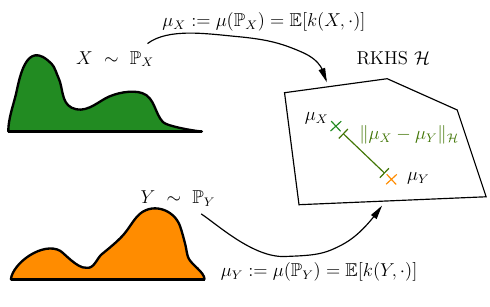}
    \caption{Illustration of the embedding of two probability distributions in the RKHS in order to compare them.}
    \label{fig:fig8}
\end{figure}
\begin{proposition}
    Given an i.i.d copy $(X_A', Y')$ of $(X_A, Y)$ such that $\mathbb{E}\left[ k(X_A, X_A')\right] < +\infty$ and $\mathbb{E}\left[ k_Y(Y, Y')\right] < +\infty$, we have :
    \begin{multline*}
        \mathrm{HSIC}(X_A,Y) = \mathbb{E} \left[ k(X_A, X'_A)k_Y(Y,Y')\right] + \mathbb{E}\left[ k(X_A, X'_A)\right]\mathbb{E}\left[ k_Y(Y, Y')\right] \\- 2\mathbb{E}\left[ \mathbb{E}\left[k(X_A, X'_A) \hspace{1pt}\vert\hspace{1pt} X_A\right] \mathbb{E}\left[ k_Y(Y,Y')\hspace{1pt} \vert\hspace{1pt} Y\right] \right]\end{multline*}
\end{proposition}
Using this simplified expression for the HSICs, we can easily estimate them using standard methods such as Monte Carlo. Unbiased (U-statistic) and biased (but asymptotically unbiased, V-statistic) estimators are introduced in \cite{arthur2005, songHSIC}. V-statistics are commonly used. To simplify notation, we will assume that $X:=X_A$.
\subsection{Estimation based on SRS}
Given two i.i.d sample $(X_i, Y_i)_{1\leq 1\leq N}$ and $(X'_i, Y_i')_{1\leq 1\leq N}$ of $(X, Y)$. The V-statistic is given by :
\begin{align*}
    \widehat{\mathrm{HSIC}}(X, Y) = \frac{1}{N^2}\sum_{i,j=1}^N k_{X}(X_i, X'_j)k_Y(Y_i, Y'_j) + \frac{1}{N^4}\sum_{i,j=1}^N k_X(X_i, X_j')\sum_{i,j=1}^N K_Y(Y_i, Y_j') \\-\frac{2}{N}\sum_{i=1}^N\left(\frac{1}{N}\sum_{j=1}^N k_X(X_i, X_j')\frac{1}{N}\sum_{j=1}^N k_Y(Y_i, Y_j')\right)
\end{align*}
Or similarly, as introduced in \cite{arthur2005} :
\begin{proposition}
    \begin{align}\label{eq:HSICtrace}
    \widehat{\mathrm{HSIC}}(X, Y) = \frac{1}{N^2}\mathrm{tr}\left(L_X H L H\right)
\end{align}
with
\begin{itemize}
    \item $L_X$ and $L$ are the Gram matrices defined as
    \[L_X = \left( k\left(X_i, X_j\right)\right)_{1 \leq i,j \leq N} \hspace{4pt} \text{and} \hspace{4pt} L = \left( k_Y\left(Y_i, Y_j\right)\right)_{1 \leq i,j \leq N}\]
    \item $H = \left(\delta_{ij} - \frac{1}{N}\right)_{1 \leq i,j \leq N}$ where $\delta_{ij}$ is the Kronecker delta.
\end{itemize}
\end{proposition}
The main advantage of this formulation is that it requires only one sample of $(X, Y)$, thanks to an i.i.d. sample. 

\subsection{Estimation based on Random Quantization}
As shown in section \ref{section:3} quantization-based LHS can be used to evaluate expectation while preserving dependency among a group of inputs. The idea is here to apply these results to the computation of HSIC. Let's define the function $f$ such that for all $(x,x')\in\mathbb{R}^2,\hspace{2pt} f(x,x') = k_X(x, x')k_Y\left(\mathcal{M}(x),\mathcal{M}(x')\right)$. 
Let $(U_i)_{i=1 \dots N}$ a sample provided by Algorithm \ref{alg:RQ} where $U_i \sim \mathcal{L}(X | X\in C_i)$, $(C_i)_{i=1,\dots, N}$ is the Voronoi tessellation associated to the quantization of $X$ and $p_i = \mathbb{P}\left[X\in C_{i}\right]$.
     \begin{proposition}
\begin{align}\label{eq:HSICRQ}
         \widehat{\mathrm{HSIC}}(X, Y) = \sum_{i,j=1}^{N}p_{i}p_{j}f(U_{i}, U_{j}) + \sum_{i,j=1}^{N}p_{i}p_{j} k_X\left(U_i, U_j\right)\sum_{i,j=1}^{N}p_{i}p_{j} k_Y\left(\mathcal{M}\left(U_i\right), \mathcal{M}\left(U_j\right)\right) \\-2\sum_{i=1}^{N}p_{i}\left[\left(\sum_{j=1}^{N}p_{j}k_X\left(U_i, U_j\right)\right)\left(\sum_{j=1}^{N}p_{j}k_Y\left(\mathcal{M}(U_i), \mathcal{M}\left(U_j\right)\right)\right)\right]
     \end{align}
\end{proposition}

\subsection{HSIC-based independence test}
The main interest of HSIC is to identify input parameters that do not affect the output. In order to obtain a distance in the RKHS, the kernels must be characteristic, i.e. injective. Therefore, the following equivalence holds for $A\subset \{1,\dots, d\}$:
\begin{align*}
    X_A \indep Y \Longleftrightarrow \mathrm{HSIC}(X_A, Y) = 0
\end{align*}
HSIC can be used to construct a statistical test of independence based on this result, introduced by \cite{gretton_kernel_2007}. The null hypothesis $\mathcal{H}_0$ : "$X_A$ and $Y$ are independent" is equivalent to $\mathrm{HSIC}(X_A,Y) = 0$. The statistic corresponding to this test is :
\begin{align*}
    \widehat{\mathcal{S}} = N\times\widehat{\mathrm{HSIC}}(X_A, Y)
\end{align*}
The p-value represents the probability that, under the null hypothesis $\mathcal{H}_0$, the observed value $\widehat{S_{obs}} = N\times \widehat{\mathrm{HSIC}}(X_A,Y)_{obs}$ is greater than $\widehat{S}$ :
\begin{align*}
\mathrm{p}_\text{val} = \mathbb{P}\left[ \widehat{\mathcal{S}} \geq \widehat{\mathcal{S}}_{\text{obs}} \hspace{1pt} \vert \hspace{1pt} \mathcal{H}_0\right]
\end{align*}
Hence, $\mathcal{H}_0$ is rejected if $\mathrm{p}_\text{val} < \alpha$, where $\alpha$ is the first order risk of the test, i.e., the risk of falsely rejecting $\mathcal{H}_0$. In practice, $\widehat{S}\mid \mathcal{H}_0$ distribution is not known. It can be approximated asymptotically to a gamma distribution (see \cite{gretton_kernel_2007}), which requires a sample size of several hundred. Alternatively, a test based on permutations and Bootstrap can be used (see \cite{de_lozzo_new_2016}).

%% file: section4.tex
In this section, we compare the Quantization-based LHS approach to Monte Carlo and LHSD on operational environmental models. The first case examines flood risk, where there is perfect knowledge of the dependency structure and the characteristics of the marginal laws. The second case studies the sizing of a grass strip in an agricultural context, where dependencies are unknown. 
\subsection{Case study I: sampling for a 1D hydro-dynamical model of flood risk}
In this first real application, the risk for an industrial site to be flooded by a river when its height exceeds a dyke is simulated with a simplified 1D- St-Venant equation \cite{Mondal_2020}. Results with LHSD and Quantization-based LHS samplings are compared to study stratified random sampling for dependent inputs when the dependence structure is well known \cite{ioossedf, chastaing2012generalized}. 

The problem consists of 8 dependent input variables, $(Q, K_{s}, Z_{v}, Z_{m}, H_{d}, C_{b}, L, B)$ summarized in Table \ref{tab:tab1}. Let us estimate $\mathbb{E}\left[S\right]$ where :
\begin{align*}
    S &= Z_{v} + H - H_{d} - C_{b}\\
    H &= \left(\frac{Q}{BK_{s}\sqrt{\frac{Z_{m} - Z_{v}}{L}}}\right)^{0.6}
\end{align*}

\begin{table}[htbp]
    \centering
    \begin{tabular}{@{}llll@{}}
        \toprule
        \textbf{Input} & \textbf{Description} & \textbf{Unit} & \textbf{Probability Distribution} \\ \midrule
        $Q$ & Maximum annual flow rate & m$^3$/s & Truncated Gumbel $G(1013, 558)$ over $[500, 3000]$ \\
        $K_s$ & Strickler coefficient & – & Truncated Normal $\mathcal{N}(30, 8)$ over $[15, \infty)$ \\
        $Z_v$ & Downstream river level & m & Triangular $T(49, 50, 51)$ \\
        $Z_m$ & Upstream river level & m & Triangular $T(54, 55, 56)$ \\
        $H_d$ & Height of the dike & m & Uniform $\mathcal{U}([7, 9])$ \\
        $C_b$ & Bank level & m & Triangular $T(55, 55.5, 56)$ \\
        $L$ & Length of the river section & m & Triangular $T(4990, 5000, 5010)$ \\
        $B$ & Width of the river & m & Triangular $T(295, 300, 305)$ \\ \bottomrule
    \end{tabular}
    \caption{Description of the model inputs.}
    \label{tab:tab1}
\end{table}

The eight inputs of the flood problem are dependent, and a Gaussian copula is proposed for the joint distribution:
\begin{align*}
C(u_1, \dots, u_{d}) = \Phi_{\text{joint}}\left(\Phi^{-1}(u_1), \dots, \Phi^{-1}(u_d)\right)
\end{align*}
where, $\Phi_{\text{joint}}$ is the cumulative distribution function of the multivariate Gaussian distribution with covariance matrix $\Sigma$, and $\Phi$, the cumulative distribution function of the standard normal distribution. In the case of the flood, dependency only exists pairwise, with the following correlation coefficients: $\rho(Q, K_s) = 0.5, \hspace{2pt} \rho(Z_v, Z_m) = \rho(L,B)=0.3$.

\begin{figure}[htbp]
    \centering
    \includegraphics[width=\textwidth]{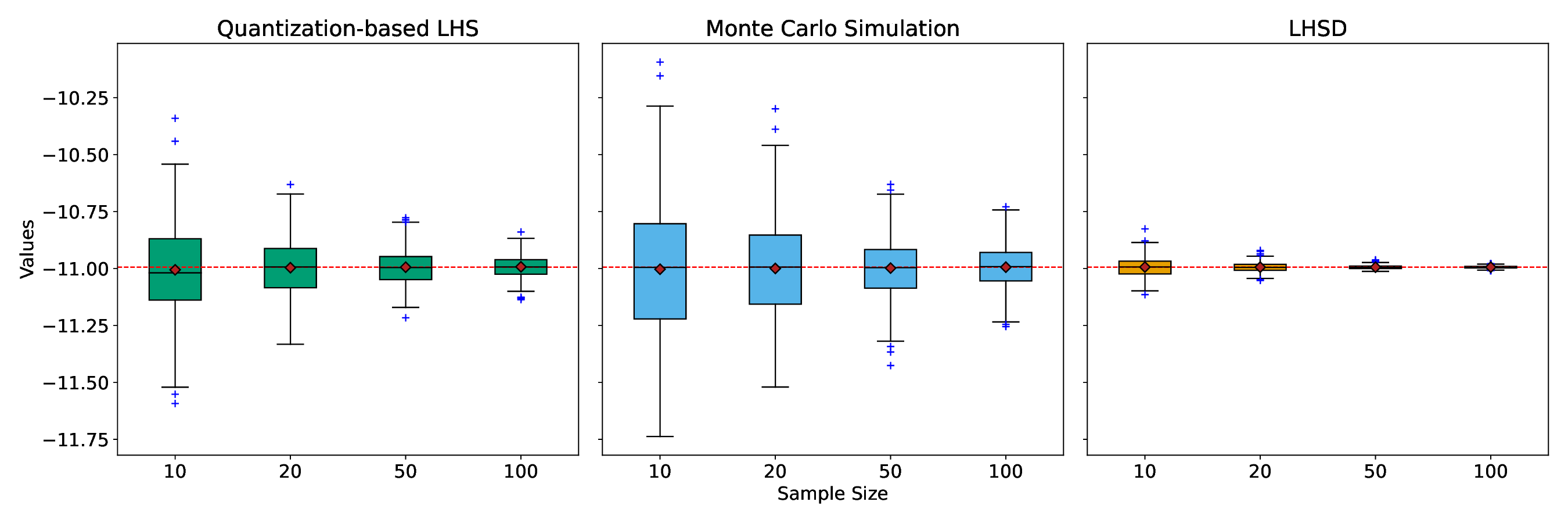}
    \caption{Estimation of $\mathbb{E}[S]$ with 500 repetitions per sample size $N \in \{10, 20, 50, 100\}$. }
    \label{fig:fig9}
\end{figure}
The results show that both LHSD and Quantization-based LHS offer better performance than Monte Carlo, although $\mu_{QLHS}$ has a higher variance than LHSD (Figure \ref{fig:fig9}). This was expected given that LHSD possesses analytical knowledge of the copula and the c.d.f and inverse c.d.f of the marginals. In environmental modeling, however, the inputs must be measured on the field or come from empirical relationships, for example. In most cases, information on the dependency structure between inputs may be completely unknown, or limited, for example coming from a random generator. In that case, the Quantization-based LHS design proposed in section \ref{sec:RQ} is particularly adapted since it only requires the application of k-means over a large sample size. In the next section, these sampling strategies are implemented and compared on a digital twin of an agricultural catchment, without any analytical knowledge of dependency.

\subsection{Case study II : sampling of Soil water retention for pesticide transfer modeling}
\subsubsection{Model and data description}
In order to reduce the river's pollution in agricultural catchments, some best management practices consist in applying vegetative filter strips (VFSs) that reduce significantly surface runoff and erosion from the cultivated fields \cite{lacas2005,reichenberger2007}. These nature-based solutions must be designed optimally to be efficient and socially accepted, considering the local conditions of soil, climate, topography, and cultural practices. To that aim, \cite{Carluer2017,veillon2022buvardmes} developed the decision-making tool BUVARD\_MES for french farmers or stakeholders in the water quality domain, based on the benchmark numerical model VFSMOD \cite{carpena1999,munoz-carpena_design_2004,laucar2018} (see figure \ref{fig:fig10}). In this study case, BUVARD\_MES is extended on the digital twin of the Morcille catchment (Figure \ref{fig:fig11}), a vineyard agricultural place in the Beaujolais region (France), where water, sediment and pesticide are intensively measured for more than 30 years \cite{Morcille_datapaper}. This digital twin, deeply tested and described in \cite{PULSE_rapport}, allows simulating transfers in fields and VFSs in all possible places of the catchment, thus running on a large sample of inputs. 

\begin{figure}[htbp]
    \centering
    \includegraphics[width=0.8\textwidth]{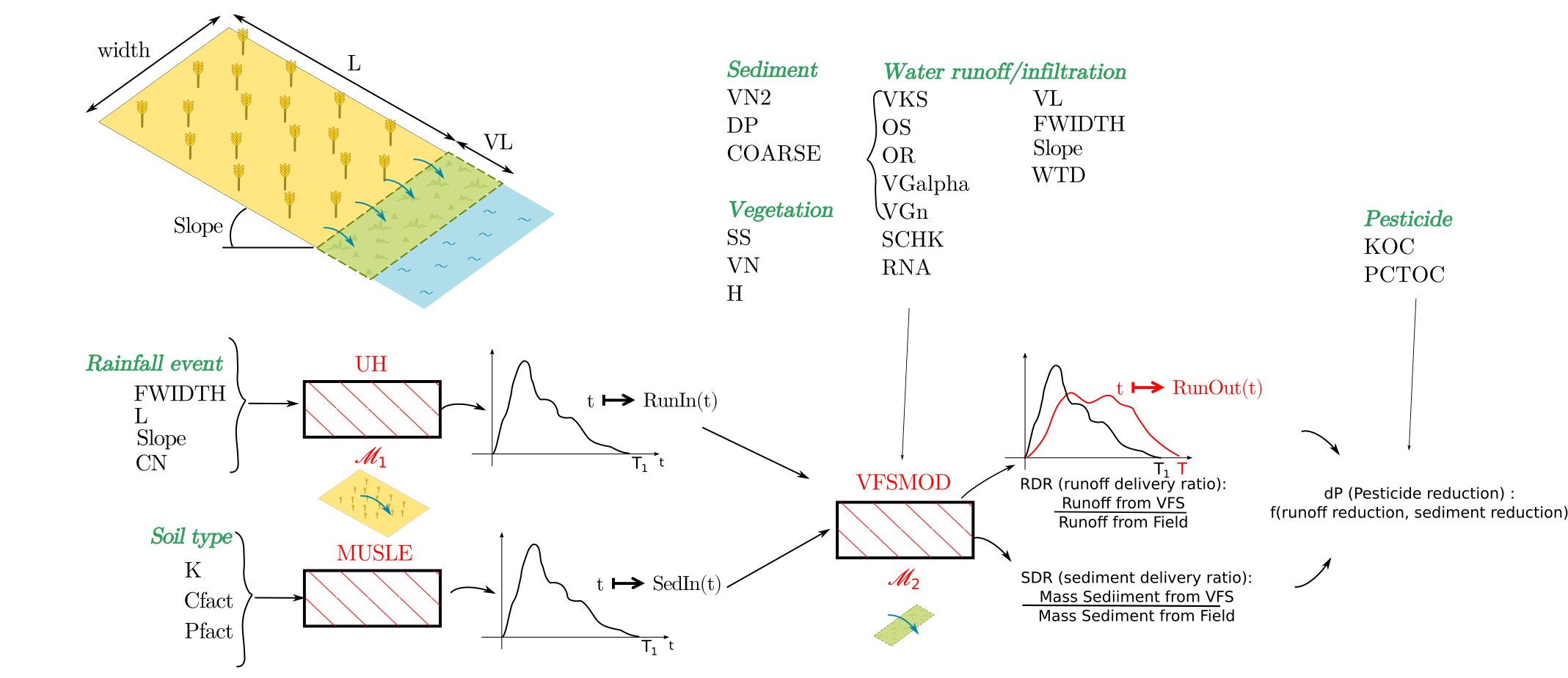}
    \caption{BUVARD\_MES model and its sub-models, with inputs for climate, soil, vegetation properties of the fields and VFSs. The group of (Van Genuchten) dependent parameters is indicated by a brace.}
    \label{fig:fig10}
\end{figure}

\begin{figure}[htbp]
    \centering
    \includegraphics[width=0.8\textwidth]{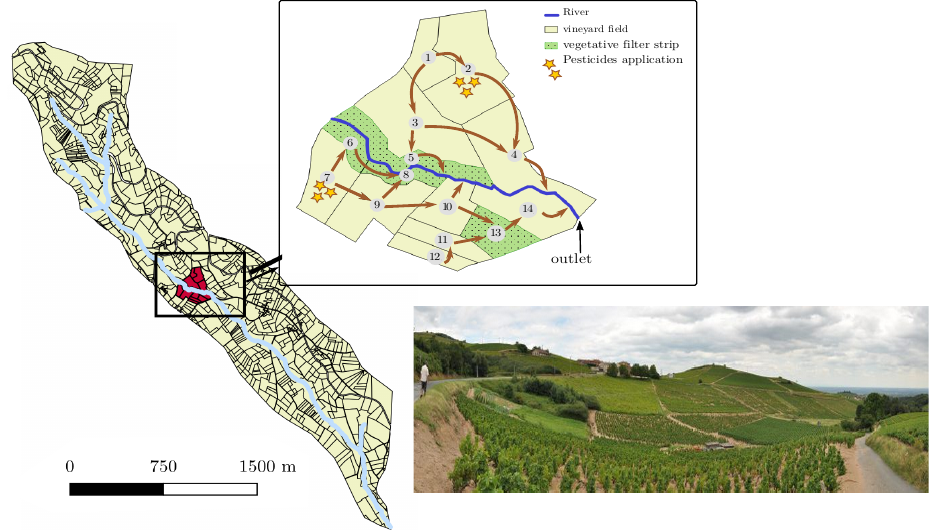}
    \caption{The Morcille catchment (left and bottom) and its digital twin, in the Beaujolais vineyard region (France). Example of surface properties extracted from the virtual catchment (top).}
    \label{fig:fig11}
\end{figure}

\subsubsection{Soil water retention estimation}
In the model, infiltration in presence of a water table is represented by the SWINGO algorithm \cite{carlaucar2018}, which depends on Van Genuchten soil hydraulic functions (VG, \cite{VG} , eq. \ref{eq:VG_theta}). 
\begin{equation}\label{eq:VG_theta}
    \theta(h) = \theta_{r} +\frac{\theta_{s} - \theta_{r}}{\left(1 + \left(\alpha\vert h\vert \right)^{n}\right)^{1 - 1/n}}
\end{equation}
where $\theta_{s}$ is the saturated water content, $\theta_{r}$ is the residual water content, $\alpha$ is linked to the inverse of the air entry suction and $n$ is related to the pore-size distribution.\\
This conductivity is described by: 
\begin{equation}
    K_v(h) = K_{\text{sat}}\sqrt{S(h)}\left(1 - \left(1 - S(h)^{\frac{1}{1 - 1/n}}\right)^{1- 1/n}\right)^2
\end{equation}
where $K_{\text{sat}}$ is the hydraulic conductivity at saturation and :
\begin{align*}
    S(h) = \frac{\theta(h) - \theta_r}{\theta_s - \theta_r}
\end{align*}

Dependencies between the soil properties in the VG equations are known to exist, but their structure is not explicitly known, despite many studies. For example, \cite{Lehmann2020} constraints the sampling by simultaneously estimating soil water characteristics and capillary length with pedotransfer functions, and \cite{Regaladocarpena2004} estimates a stochastic relation between some of the VG parameters on some specific soils. 
In order to account for this unknown information, two properties are considered to describe the inputs in BUVARD\_MES for the sampling: the first set of inputs consider them  as independent and thus random, and the second set is made of dependent variables (the Van Genuchten set, 5 parameters). For this set, a random generator was used on the data to generate the joint distribution.\\ 

Figure \ref{fig:fig12} illustrates $\theta(h)$ curves with Random Quantization  (Algorithm \ref{alg:RQ}) and LHS sampling, clearly showing the values of $\theta_r$ (minimum water content) and $\theta_s$ (maximum water content), which correspond to physical values and exhibit a trend consistent with reality. The LHS curves, however, are not in agreement with observed physical behavior.\\ 

\begin{figure}[htbp]
    \centering
    \includegraphics[width=0.9\textwidth]{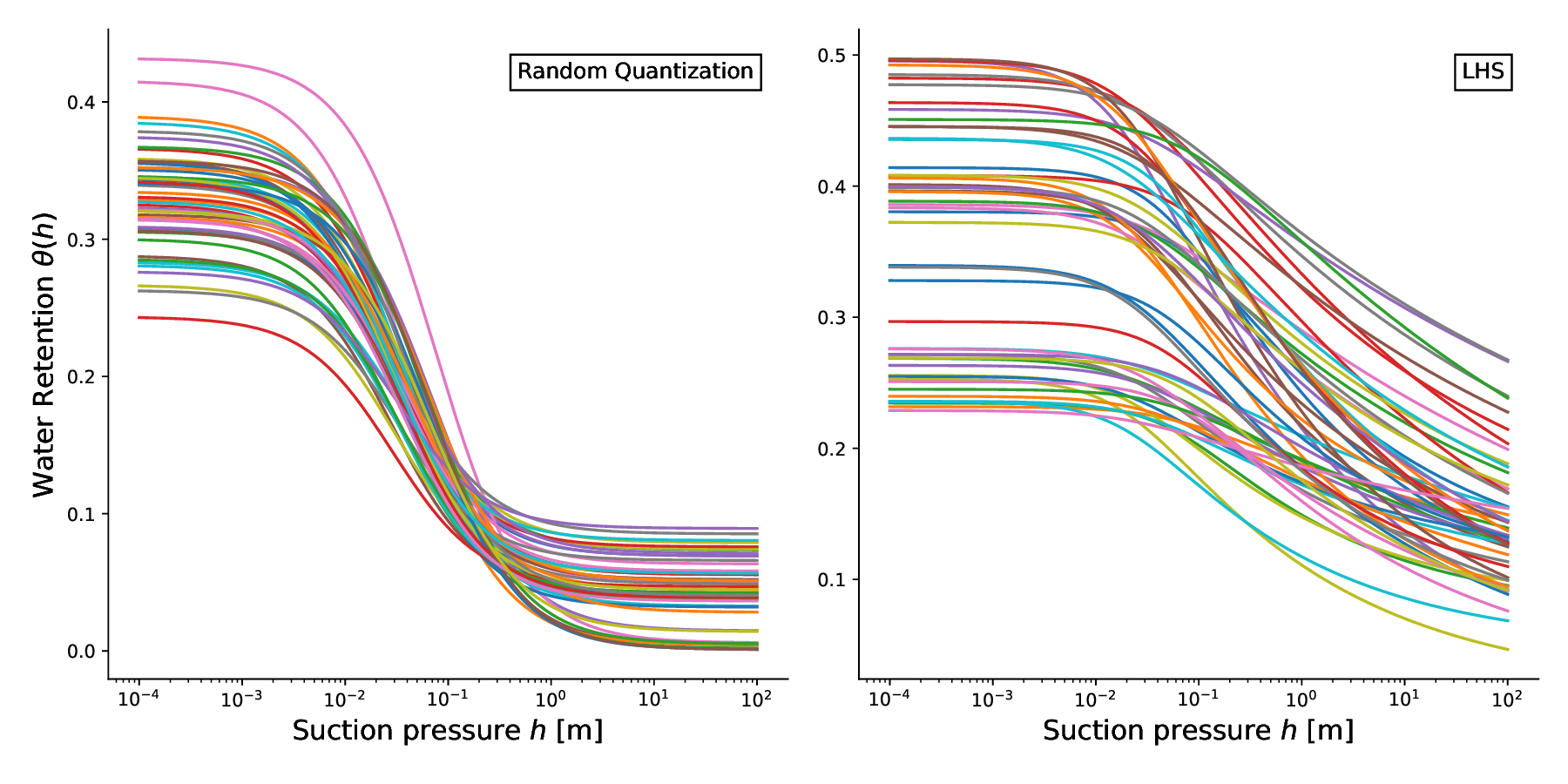}
    \caption{Water retention curves $\theta(h)$ for different Van Genuchten parameters samples based on Random Quantization (left) and LHS (right).}
    \label{fig:fig12}
\end{figure}

For the LHSD case, a Gaussian copula was fitted using maximum likelihood. As we do not have access to the quantile and cumulative distribution functions, we used their empirical versions. Results given in Figure \ref{fig:fig13} show that the Random Quantization estimator is unbiased with lower variance than Monte Carlo and LHSD. This confirms the relevance of this method, considering the simplicity of implementing this approach compared to LHSD. Indeed, it only requires a simple k-means, while LHSD requires estimating a good copula and distribution function estimates, which can be a time-consuming process. Figure \ref{fig:fig14} shows another illustration on the conductivity curve. Despite the difficulty of the problem, which lies in the small value to be estimated, all three methods provide satisfactory results. $\mu_{RQ}$ outperforms both LHSD and Monte Carlo by providing a lower variance unbiased estimate.

\begin{figure}[htbp]
    \centering
    \includegraphics[width=\textwidth]{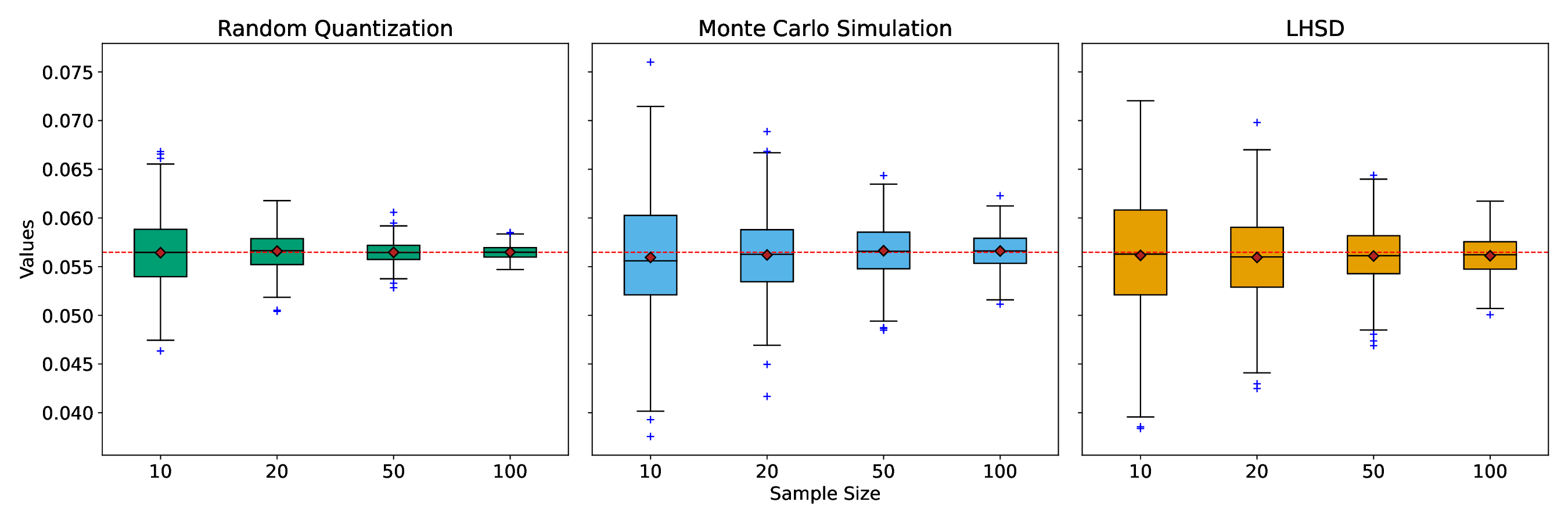}
    \caption{Comparison of Random Quantization, Monte Carlo, and LHSD on the Water content for $h = 1$m for sample sizes $N\in\{10, 20, 50, 100\}$, with 500 replicates per sample size. The LHSD was modelled using a Gaussian copula and estimated through maximum likelihood and empirical quantile function.}
    \label{fig:fig13}
\end{figure}

\begin{figure}[htbp] 
    \centering
    \includegraphics[width=\textwidth]{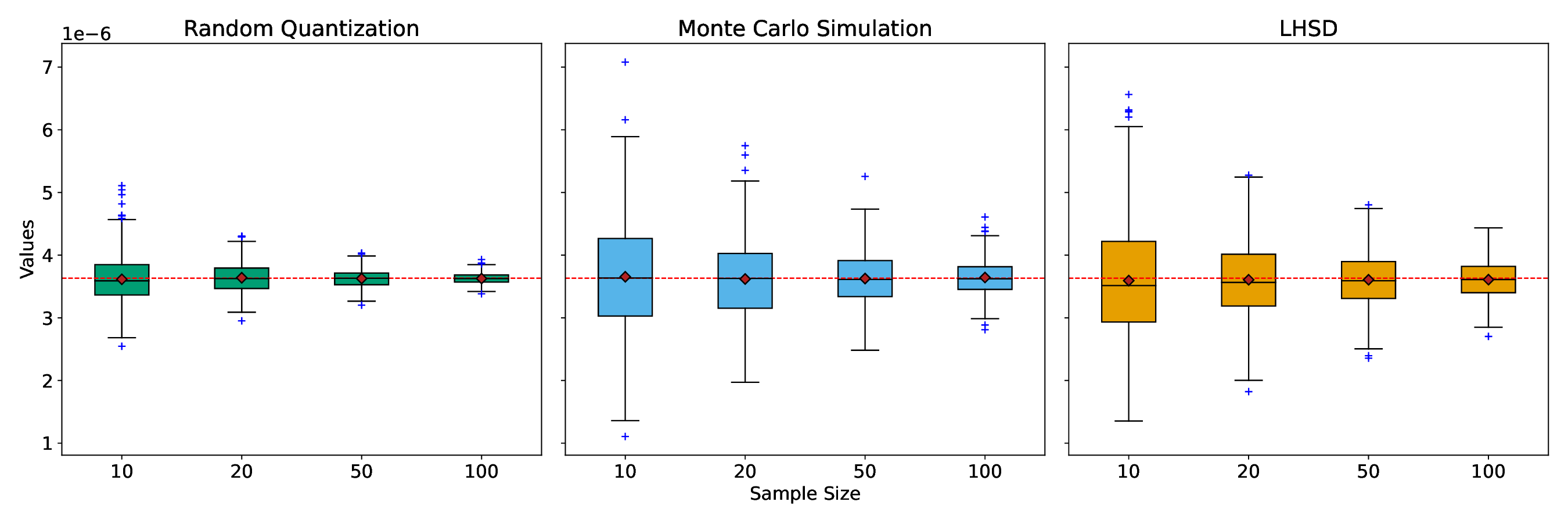}
    \caption{Comparison of Random Quantization, Monte Carlo, and LHSD on the conductivity curve estimation for $h=10^{-3}$ was conducted using sample sizes $N\in\{10, 20, 50, 100\}$, with 500 replicates per sample size. The LHSD was modelled using a Gaussian copula and estimated through maximum likelihood and empirical quantile function.}
    \label{fig:fig14}
\end{figure}

Finally, besides a point estimate in $h$ to illustrate the methods efficiency, the average water retention curve $\theta(h)$ was estimated on the whole range of values of pressure for all sampling sizes (Figure \ref{fig:fig15}). We observe a regularity in the estimation quality over the suction pressure. The $N=10$ case is highlighted in Figure \ref{fig:fig16}. The best results were obtained through random quantization sampling, showing the robustness of the method on the whole range.\\

\begin{figure}[htbp]
    \centering
    \includegraphics[width=0.95\textwidth]{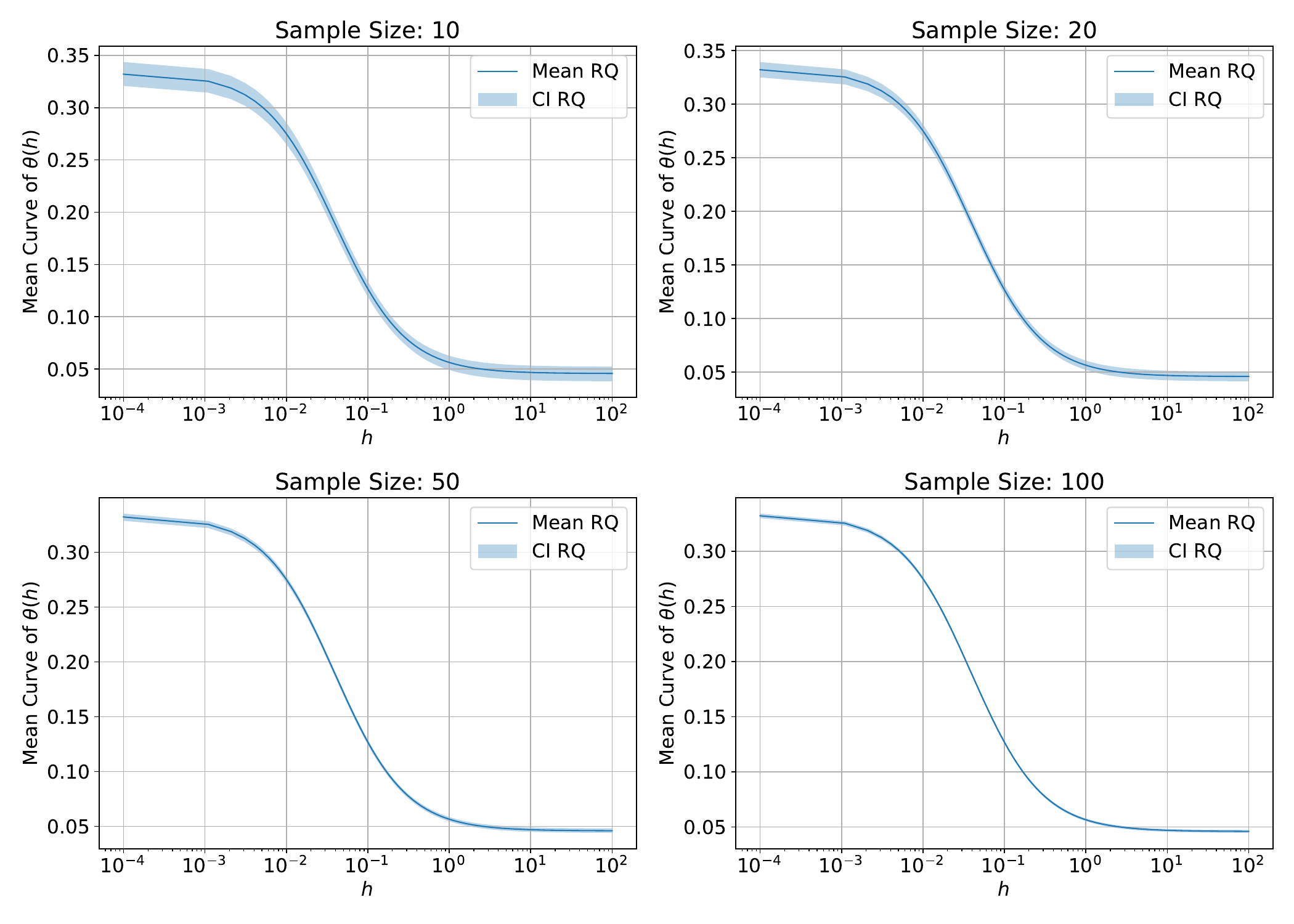}
    \caption{Mean Water retention curve $\theta(h)$ for $h\in [10^{-4}, 10^2]$ with Random Quantization sampling for different sample sizes $N \in \{10, 20, 50, 100\}$. The confidence region corresponds to the 2.5\%-percentile and 97.5\%-percentile.}
    \label{fig:fig15}
\end{figure}

\begin{figure}[htbp]
    \centering
    \includegraphics[width=0.55\textwidth]{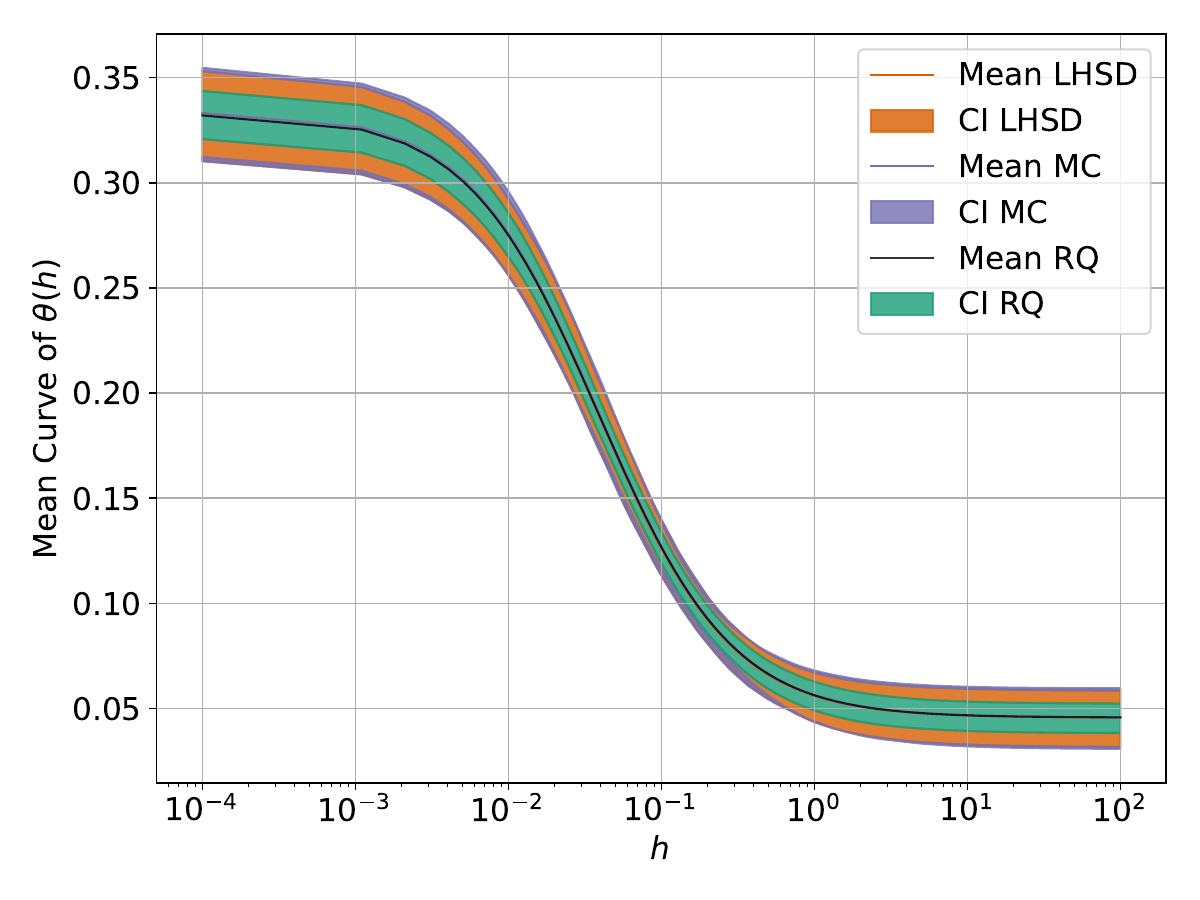}
    \caption{Mean Water retention curve $\theta(h)$ for $h\in [10^{-4}, 10^2]$ with Random Quantization, LHSD and Monte Carlo sampling with $N=10$. The confidence region corresponds to the 2.5\%-percentile and 97.5\% percentile.}
    \label{fig:fig16}
\end{figure}

\subsubsection{Sensitivity analysis with HSIC}

To reduce the complexity of the BUVARD-MES model, which is time-consuming to compute, it is necessary to reduce the dimensions in order to keep only those that influence the output. To address this issue, HSIC independence tests are performed as described in section  \ref{section:4}. The group of dependent variables is considered as a single input, referred to as Van Genuchten. For LHSD, the same parameterization as in the previous example is maintained.\\

In practice, the analysis included 10 uncorrelated inputs that described the geometric properties of the contributing surface (CA) and the VFS (Area, Length, Slope, Curve Number of the field ; Slope, Width, Organic Matter, Clay content and Water table depth of the VFS ; and the pesticide property Koc, see Table \ref{tab:tab2}), as well as properties related to pesticides and organic matter. Additionally, the group of 5 correlated Van Genuchten inputs were used to describe soil conductivity and water retention capacity. \\

The independent inputs are associated with a univariate Radial Basis Function (RBF) kernel, while the dependent group is associated with an adapted RBF kernel. To ensure consistency in the parameterization of this kernel, data were standardized, as the use of a standard deviation for the whole group is restrictive due to the 5 variables having a very variable order of magnitude (ranging from $10^{-5}$ to $10^1$). As the output is scalar, a univariate RBF is used. If the p-value is less than 5\%, $\mathcal{H}_0$ is rejected. Otherwise, $\mathcal{H}_0$ is accepted, i.e. the considered input is independent of the output.
\begin{align*}
    \forall(x, x') \in\mathbb{R}^2,\quad k_{\text{indep}}(x, x') = \exp\left(-\frac{(x-x')^2}{2\theta^2}\right), \quad \theta\in\mathbb{R}\\
    \forall(x, x') \in\left(\mathbb{R}^d\right)^2,\quad k_{\text{dep}}(x, x') = \exp\left(-\frac{\Vert x-x'\Vert^2}{2\theta^2}\right), \quad \theta\in\mathbb{R}
\end{align*}

The results of this independence test are summarized in Table \ref{tab:tab2}. The reference values have been obtained by an asymptotic test using a gamma distribution and a Monte Carlo draw of 10,000 points. For other methods, p-values are obtained by bootstrap. The results between the reference and the proposed Quantization-based LHS method are in agreement despite the small sample size. Furthermore, by using either the LHSD or Monte Carlo approach (with 400 points), we can accept the hypothesis that Van Genuchten is independent of runoff efficiency. This contradicts the reference and 'expert' knowledge. 

\begin{table}[htbp]
  \centering
  \begin{tabular}{lccccccccc}
    \toprule
    \multirow{2}{*}{Input} & \multicolumn{4}{c}{p-value} 
      & \multicolumn{4}{c}{Decision} \\
    \cmidrule(lr){2-5}
    \cmidrule(lr){6-9}
      & Ref & MC & QLHS & LHSD & Ref & MC & QLHS & LHSD \\
    \midrule
    Area  CA & $0$ & $0$ &$0$& $0$ & \ding{51} & \ding{51} & \ding{51} & \ding{51} \\
    Length CA & $0.38$ & $0.40$ &$0.86$& $0.75$ & \ding{55} & \ding{55} & \ding{55} & \ding{55} \\
    Width CA  & $0$ & $0$ &$0$& $0$ & \ding{51} & \ding{51} & \ding{51} & \ding{51} \\
    Slope CA & $0.54$ & $0.032$ &$0.33$& $0.87$ & \ding{55} & \ding{51} & \ding{55} & \ding{55} \\
    Slope VFS  & $0$ & $0.012$ & $0.0018$& $0.0004$ & \ding{51} & \ding{51} & \ding{51} & \ding{51} \\
    Width VFS  & $0$ & $0$ &$0$& $0$ & \ding{51} & \ding{51} & \ding{51} & \ding{51} \\
    OM VFS & $0.27$ & $0.81$ &$0.87$& $0.23$ & \ding{55} & \ding{55} & \ding{55} & \ding{55} \\
    WTD  VFS & $0$ & $0$ &$0$& $0$ & \ding{51} & \ding{51} & \ding{51} & \ding{51} \\
    C VFS & $0.87$ & $0.09$ &$0.29$& $0.85$ & \ding{55} & \ding{55} & \ding{55} & \ding{55} \\
    Koc & $0.38$ & $0.09$ &$0.67$& $0.89$ & \ding{55} & \ding{55} & \ding{55} & \ding{55} \\
    CN CA& $0$ & $0$ &$0$& $0$ & \ding{51} & \ding{51} & \ding{51} & \ding{51} \\
    Van Genuchten  & $0$ & $0.18$ &$0.031$ & $0.096$ & \ding{51} & \ding{55} & \ding{51} & \ding{55} \\
    \bottomrule
  \end{tabular}
  \caption{HSIC independence test on BUVARD-MES with 400 points per sample method. 'CA' stands for the contributive area of the VFS (the field), 'VFS' stands for vegetative filter strip.  WTD is Water Table depth, C is clay content, Koc is the pesticide soil adsorption coefficient, CN is the Curve Number. The HSIC for MC and LHSD were computed with Equation \ref{eq:HSICtrace}. For QLHS, with Equation \ref{eq:HSICRQ}. \ding{51} : The output is dependent of the input. \ding{55} : The output is independent of the input.}
  \label{tab:tab2}
\end{table}

%% file: conclusion.tex
This article proposes a new space-filling experimental design, called Quantization-based LHS, that incorporates dependency. The sampling is based on vector quantization, specifically k-means, ensuring ease of implementation while incorporating any structure of dependence. The DOE is built in an LHS way, ensuring comprehensive coverage of each marginal including groups of dependent inputs, and requires few evaluations. It allows for unbiased estimation of expectations in various configurations. 
The methodology has been applied to several case studies, including HSIC kernel sensitivity analysis. We show that the use of Quantization-based LHS allows for high-performance sensitivity analysis with a smaller sample size compared to existing sampling approaches.\\

Consequently, on the basis of this methodology, and while ensuring that the dependency structure of the inputs is taken into account, a screening step can be carried out that allows the input dimension to be reduced in order to limit the calls to the computational code and to build an accurate metamodel.